\documentclass[10pt,journal]{IEEEtran}
\IEEEoverridecommandlockouts
\usepackage{amsmath}
\usepackage{textcomp}
\usepackage[colorlinks, linkcolor=black,
anchorcolor=black,citecolor=black]{hyperref}
\usepackage{makecell}
\usepackage{bm}
\usepackage{amsmath,amssymb,amsfonts}
\usepackage{booktabs}
\usepackage{array}
\usepackage{caption}
\usepackage{graphicx}
\usepackage{algorithm}
\usepackage{algorithmic}
\usepackage{pifont}
\usepackage{harpoon}
\usepackage{multirow}
\usepackage{algorithm,float}
\usepackage{verbatim}
\usepackage{amsmath}
\usepackage{amsthm}
\newtheorem{theorem}{Theorem}

\usepackage{xcolor}
\usepackage{color}
\usepackage{subfigure}
\usepackage{threeparttable}
\usepackage{diagbox}
\usepackage{bbding}
\usepackage{pifont}
\usepackage{wasysym}
\usepackage{enumerate}
\usepackage{enumitem}
\usepackage{eqparbox}

\def\BibTeX{{\rm B\kern-.05em{\sc i\kern-.025em b}\kern-.08em
    T\kern-.1667em\lower.7ex\hbox{E}\kern-.125emX}}
\begin{document}

\title{Privacy-preserving Decentralized Deep Learning with
Multiparty Homomorphic Encryption}

\author{Guowen~Xu, Guanlin~Li, Shangwei~Guo, Tianwei~Zhang, Hongwei~Li,~\IEEEmembership{Senior Member,~IEEE}

\IEEEcompsocitemizethanks { \IEEEcompsocthanksitem Guowen~Xu, Guanlin~Li, and Tianwei~Zhang are with the School of Computer Science and Engineering, Nanyang Technological University. (e-mail: guowen.xu@ntu.edu.sg; guanlin001@e.ntu.edu.sg; tianwei.zhang@ntu.edu.sg)
\IEEEcompsocthanksitem Shangwei~Guo is  with the  College of Computer Science, Chongqing University, Chongqing 400044, China.(e-mail: swguo@cqu.edu.cn)
\IEEEcompsocthanksitem  Hongwei~Li is with the School of Computer Science and Engineering,  University of
Electronic Science and Technology of China, Chengdu 611731, China. (e-mail: hongweili@uestc.edu.cn) }}

 \IEEEcompsoctitleabstractindextext{
\begin{abstract}
\renewcommand{\raggedright}{\leftskip=0pt \rightskip=0pt plus 0cm}
 \raggedright
Decentralized deep learning plays a key role in collaborative model training due to its attractive properties, including tolerating high network latency and less prone to single-point failures. Unfortunately, such a training mode is more vulnerable to data privacy leaks compared to other distributed training frameworks. Existing efforts exclusively use differential privacy as the cornerstone to alleviate the data privacy threat. However,  it is still not clear whether differential privacy can provide a satisfactory utility-privacy trade-off for model training, due to its inherent contradictions. To address this problem, we propose \textbf{D$^{2}$-MHE}, the \textit{first} secure and efficient  decentralized training framework with lossless precision. Inspired by the latest developments in the homomorphic encryption technology, we design a multiparty version of Brakerski-Fan-Vercauteren (BFV), one of the most advanced cryptosystems, and  use it to implement private gradient updates of users' local models. \textbf{D$^{2}$-MHE} can reduce the communication complexity of general Secure Multiparty Computation (MPC) tasks from quadratic to linear in the number of users, making it very suitable and scalable for large-scale decentralized learning systems.
   Moreover, \textbf{D$^{2}$-MHE} provides  strict semantic security protection even if the  majority of users are dishonest with collusion. We conduct extensive experiments on  MNIST and CIFAR-10 datasets to demonstrate the superiority of \textbf{D$^{2}$-MHE}  in terms of model accuracy, computation and communication cost compared with existing schemes.

\end{abstract}
}
\maketitle

\IEEEdisplaynotcompsoctitleabstractindextext

\IEEEpeerreviewmaketitle

\section{Introduction}
As a promising technology, deep learning has been widely used in various scenarios ranging from autonomous driving, video surveillance  to face recognition. To achieve satisfactory performance for complex artificial intelligence tasks, modern deep learning models need to be trained from excessive computing resources and data samples. Hence, a conventional approach is centralized training (Figure ~\ref{Fig:Different types of training frameworks}(a)): each user is required to upload their training samples to a third party (e.g., cloud server), which has enough computing resources to produce the final model. However, this fashion raises widespread privacy concerns about the training data \cite{sav2020poseidon}. Intuitively, an untrusted third party has financial incentive to abuse the sensitive data collected from different users, such as malicious dissemination, packaging and selling them to the black market.

\begin{figure*}
  \centering
  \subfigure[]{\includegraphics[width=0.3\textwidth]{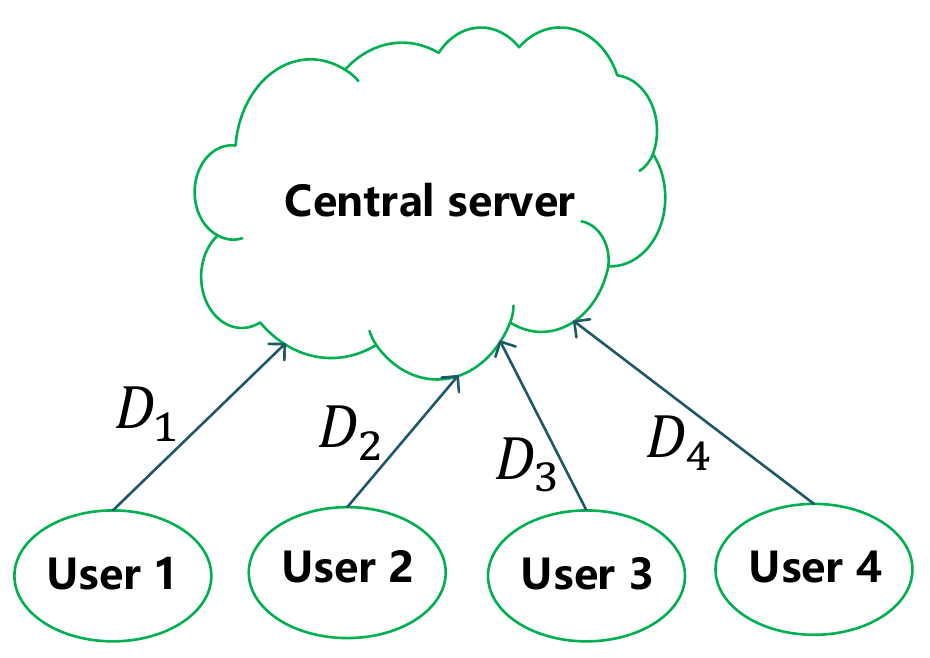}}
  \subfigure[]{\includegraphics[width=0.3\textwidth]{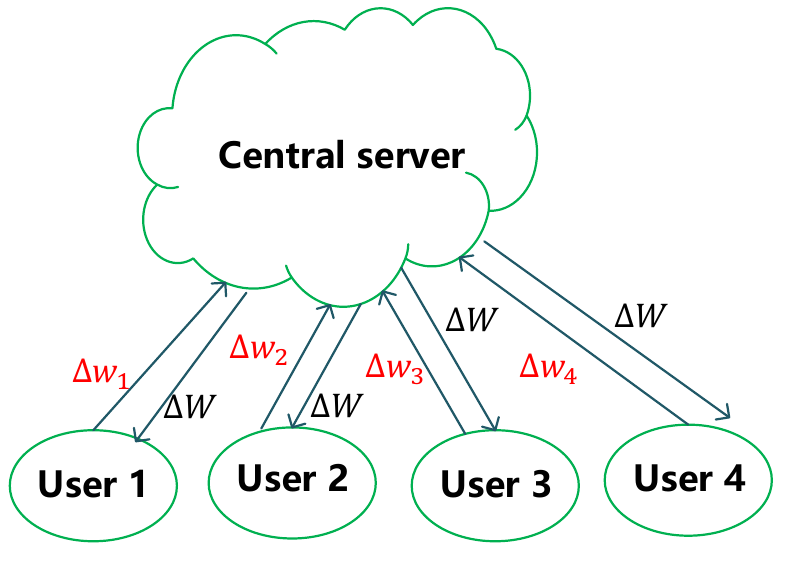}}
  \subfigure[]{\includegraphics[width=0.3\textwidth]{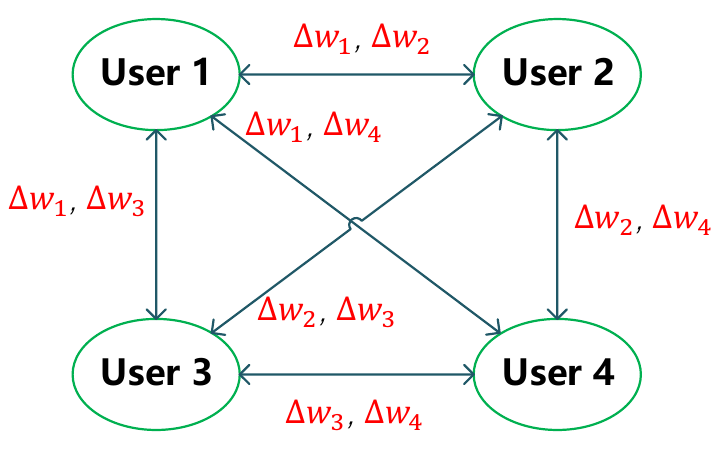}}
  \caption{Different types of training systems. (a) Centralized learning: \textnormal{ each user $i$ uploads its dataset $D_i$ to a central server, which trains a specific model in a centralized manner}.  (b) Federated learning: \textnormal{ each user train a local model with its own dataset. A central server is introduced to aggregate the gradients $\Delta w_i$ uploaded by each user $i$. Then, each user updates the local model using the aggregated value (global parameter $\Delta W$) returned by the server.}  (c) Decentralized learning: \textnormal{each user trains its own local model, and exchanges gradients with other users interconnected with it in the network. Meanwhile, it collects the gradients from the neighbors, aggregates them and updates its local model.}}
  \label{Fig:Different types of training frameworks}
\end{figure*}

To alleviate the above problem, one potential way is to split and distribute the training task to the users, who only need to train the models locally and then share the gradients without disclosing their private data.  As a result, untrusted third parties do not have access to the users' data, and the privacy risk of training data is effectively reduced. For instance, in federated learning, a central server collects the gradients from all these users, aggregates them and distributes the new gradient to each user (Figure ~\ref{Fig:Different types of training frameworks}(b)). In a decentralized learning system, the central server is eliminated, so each user autonomously exchanges gradients with its neighbors and updates its model (Figure~\ref{Fig:Different types of training frameworks}(c)).

Even though the users do not release their training samples in these distributed systems, the adversary can still infer the attributes of these samples  or even preciously reconstruct the original samples \cite{LuoWXO21,NEURIPS2019} from the shared gradients.
This threat is more severe in decentralized learning than federated learning \cite{cheng2019towards,xiao2019local}, as any user connected to an honest user can access its gradients and compromise its data privacy, making the potential attack domains and methods more diversified and concealed.
Since decentralized learning exhibits higher robustness to network delays and single point of failures, it becomes more promising and suitable for distributed training with large numbers of users. Therefore, it is necessary to have an efficient and privacy-preserving decentralized deep learning framework. Unfortunately, privacy protection of decentralized learning systems is still in its infancy. Although a wealth of works have been proposed to protect the privacy of federated learning \cite{WuCXCO20,sav2020poseidon,zheng2019helen,froelicher2020scalable}, they cannot be easily extended to the decentralized learning scenario due to its unique network topology and gradient propagation mechanism.


In particular, existing privacy-preserving deep learning solutions are mainly evolved from the following three technologies. (1) \textit{Differential Privacy} \cite{fc4ddc15,Li2020Differentially}: this approach adds controllable noise to the users' data, gradients or intermediate values to obfuscate the adversary's observations while maintaining the training accuracy. However, it  is  still  unclear  whether  differential  neural network  training  can  provide  a  satisfactory utility-privacy trade-off in practical scenarios. (2) \textit{Secure Multi-Party Computation} (MPC) \cite{kumar2020cryptflow,mohassel2018aby3}: this approach enables multiple entities to securely compute arbitrary functions without revealing their secret inputs. It has been widely used in centralized learning and federated learning systems \cite{DChaudhariRS20,agrawal2019quotient,patra2020blaze,rathee2020cryptflow2}. However, it is hard to be grafted to the decentralized scenario due to the lack of central servers, and executing the secret sharing protocols across users is rather inefficient. (3) \textit{Homomorphic Encryption} (HE) \cite{zheng2019helen,sav2020poseidon}: this approach enables the calculation of arbitrary (approximate) polynomial functions in a ciphertext environment without the need for decryption. It has been widely used in private deep learning \cite{zhanggala,chen2019efficient,sav2020poseidon,zheng2019helen}. However, it requires expensive calculations for function evaluation under ciphertext, which can significantly affect the efficiency of decentralized learning. More analysis about the limitations of these solutions is given in Section \ref{Related Works and  Challenges}.

Driven by the above limitations, our goal is to remedy the gap in the practicality of decentralized learning for protecting the training data privacy. We propose \textbf{D$^{2}$-MHE},  a practical, privacy-preserving and high-fidelity decentralized training framework. \textbf{D$^{2}$-MHE} is built based on the HE technology with innovations to address the computational bottleneck of ciphertext operation and distributed decryption. We explore the benefits of a state-of-the-art HE method, Brakerski-Fan-Vercauteren (BFV) \cite{fan2012somewhat}, and extend it to a multiparty version for privacy-preserving decentralized learning.

Specifically, the BFV cryptosystem is a fully homomorphic encryption scheme based on the Ring-learning with error (RLWE)  problem. It supports both addition and multiplication operations in ciphertext. Compared with the standard BFV, the main difference of our multiparty version  is that the decryption capability is divided into $N$ users. It means that the public key $pk$ used for encryption is disclosed to all the users, while the secret key $sk$ is  divided into $N$ shares and can only be recovered with the collaboration of these $N$ users. As a result, to construct such a multiparty version, all the algorithms requiring $sk$ as input need to be modified to meet the needs of distributed decryption. In detail,  we construct four new functions ($\mathsf{MBFV\cdot SecKeyGen}$, $\mathsf{MBFV \cdot PubKeyGen}$, $\mathsf{MBFV\cdot Bootstrap}$ and $\mathsf{MBFV\cdot Convert}$)  based on the standard BFV cryptosystem, to support system secret key generation, public key generation, distributed bootstrapping,  and ciphertext conversion in a decentralized environment, respectively (See Section~\ref{Multiparty BFV for D-PSGD} for more details). These four novel functions in \textbf{D$^{2}$-MHE} can satisfy the following properties.

First, all these constructed functions are bound to the
given NP-hard problem,  to ensure the semantic security of the \textbf{D$^{2}$-MHE} cryptosystem. Second,
the BFV cryptosystem reduces the security of the scheme to the famous NP-hard problem (i.e., Decision-RLWE \cite{cheon2017homomorphic})  by adding controllable noise in the ciphertext. The noise needs to be erased in the decryption process to ensure the correctness of the decryption.  \textbf{D$^{2}$-MHE} follows such security guidelines, but accumulates more noise in the process of generating the public key  $pk$  (See Section~\ref{Multiparty BFV for D-PSGD}). Moreover, this accumulated noise will be transferred to other operations requiring $pk$ as the input. Therefore, to ensure the correctness of the decryption, we carefully control the scale of the noise added to the newly constructed functions. Third,
in the standard BFV, the decryption is performed by a party holding the secret key. However, in \textbf{D$^{2}$-MHE},  this must be done without revealing $sk$.  Obviously, once $sk$ is revealed, all the local gradients that users previously encrypted with $pk$ will be leaked. To achieve this, we design a new method to realize ciphertext conversion \cite{ducas2015fhew} in BFV, i.e., converting a ciphertext  originally encrypted under the system public key $pk$ into a new ciphertext  under the recipient's public key $pk'$.


To the best of our knowledge, \textbf{D$^{2}$-MHE} is the first work to accelerate the performance of decentralized learning by using cryptographic primitives. It provides the best accuracy-performance trade-off compared to existing works. Our contributions can be summarized as follows:

\begin{itemize}[leftmargin=*]
\item  We design a novel decentralized training framework \textbf{D$^{2}$-MHE} with the multiparty homomorphic encryption. Compared with  existing works, it reduces the communication overhead of each round of gradient update from quadratic to linear in the number of users without sacrificing the accuracy of the original model.

\item We provide a rigorous security proof for \textbf{D$^{2}$-MHE}. Theoretical analysis shows that \textbf{D$^{2}$-MHE} can provide semantic security even if most of the users participating in the training are dishonest and collude with each other.
\item We conduct extensive experiments on  MNIST and CIFAR-10 to demonstrate the  the superiority of \textbf{D$^{2}$-MHE} in performance, and the advantages of communication and computation overhead compared with existing similar schemes.

\end{itemize}

The remainder of this paper is organized as follows. Section \ref{Related Works and  Challenges} discusses related works about privacy-preserving solutions and limitations. In  Section \ref{sec:PROBLEM STATEMENT}, we review some basic concepts and introduce the scenarios and threat models in this paper. In  Section \ref{subsub:PROPOSED SCHEME}, we give the details of our  \textbf{D$^{2}$-MHE}.  Security analysis and performance evaluation are presented in Sections \ref{Security Analysis} and \ref{sec:PERFORMANCE EVALUATION}, respectively. Section \ref{sec:conclusion} concludes the paper.

\section{Related Works}
\label{Related Works and  Challenges}
We review the existing privacy-preserving solutions for deep learning, which can be classified into three categories.

\subsection{Differential Privacy (DP)}
Differential privacy mainly relies on adding controllable noise to the user's local data, gradient or intermediate value, to realize the confusion of user data but ensure the training performance \cite{yu2019differentially,mckenna2019graphical}. Several works \cite{BernauEGKK21,cheng2019towards,xiao2019local, bellet2018personalized,cheng2018leasgd,xu2020dp}  have been designed for decentralized training scenarios. For example, \textit{Cheng et al.}\cite{cheng2019towards} propose  LEASGD, which achieves a predetermined privacy budget by adding random noise to the users' local gradients, and calibrates the scaling of noise by analyzing the sensitivity of the update function in the algorithm.  \textit{Bellet et al. } \cite{bellet2018personalized} also design a completely decentralized algorithm to solve the problem of personalized optimization, and use differential privacy to protect the privacy of user data. Other works, like A$(DP)^{2}$SGD \cite{xu2020dp} and ADMM \cite{xiao2019local},  implement the perturbation of each user's gradients with similar tricks.

\textbf{Limitations}: It is still unclear whether differential neural network training can provide a satisfactory utility-privacy trade-off for common models. This stems from the inherent shortcomings of differential privacy: achieving a strong privacy protection level requires to inject a large amount of noise during the model training, which inevitably reduces the model accuracy (See Section\ref{sec:Accuracy}). Besides,  the state-of-the-art results \cite{jayaraman2019evaluating, hitaj2017deep} also show that the current differential privacy technology is only effective for simple machine learning models, and can rarely provide satisfactory accuracy-privacy trade-off for practical neural network models.

\subsection{Secure Multi-Party Computation (MPC)}
MPC allows multiple participants to securely compute arbitrary functions without releasing their secret inputs \cite{damgaard2019new,riazi2019xonn}.  It has been widely used in conventional deep learning scenarios, including centralized learning and federated learning \cite{DChaudhariRS20,agrawal2019quotient,patra2020blaze,rathee2020cryptflow2}.  Most of these efforts rely on users to secretly share (utilizing Shamir's Secret-Sharing  \cite{benhamouda2021local} or Additive Secret-Sharing \cite{mouchet2020multiparty}) local data or gradients to two or more servers, and require an honest majority to perform deep learning training and prediction without collusion. In this way, frequent secret sharing between users is avoided, and the complexity of communication overhead is reduced from O$(N^2)$ to O($S^{2}$), where $N$ and $S$ represent the numbers of users and servers, respectively.

\textbf{Limitations}:  It is convincing to explore MPC-based protocols under centralized or federal learning, because third-party servers naturally exist in these scenarios. However, grafting MPC to a decentralized scenario has the following limitations.
(1) Decentralized learning abandons the central servers to avoid single point of failure and communication bottlenecks. As a result, it is conflicting to transplant the existing MPC-based training mechanism to a decentralized mode.  A trivial idea to alleviate this problem is to
execute the secret sharing protocol between users directly, which is rather inefficient as each user needs to perform $N-1$ interactions for secret sharing at each iteration \cite{NEURIPS2020_5bf8aaef} (refer to Section~\ref{sec:Communication Overhead}).
(2) The existing MPC technology generally requires that most of the entities involved in the calculation are honest and will not collude with each other \cite{DChaudhariRS20,agrawal2019quotient,patra2020blaze,rathee2020cryptflow2}. This is to ensure the smooth execution of calculations.  In other words, if a majority of entities are dishonest and collude with each other, there is a high probability that execution will be terminated or errors will occur.  However, a strong security framework should be able to withstand attacks from adversarial collusion. In a decentralized scenario, the need for such security guarantee is more urgent, because any user can obtain the gradient of other users connected to it, and then easily collude with some malicious users to break the privacy of the target user.

\subsection{Homomorphic Encryption (HE)}
 (Fully) homomorphic encryption can achieve the calculation of arbitrary (approximate) polynomial functions under ciphertext without the need for decryption \cite{chillotti2020tfhe,chen2019improved}. Such attractive nature makes it widely used in private deep learning \cite{zhanggala,chen2019efficient,sav2020poseidon,zheng2019helen}. Informally, we can divide HE into the following two types with different decryption methods: (1) standard HE \cite{van2010fully,gentry2012fully} is mainly used for model inference, where the public key is released to all the participants, while the secret key is only held by the decryptor (e.g., the user).  (2) In threshold-based HE \cite{chen2019multi,genise2019homomorphic,cramer2001multiparty}, the secret key is securely shared with multiple entities. As a result, each entity still performs function evaluation under the same public key, while the decryption of the result requires the participation of the number of entities exceeding the threshold. Several threshold-based HE variants \cite{cramer2001multiparty,chen2019multi, boneh2018threshold} have been used in the federated learning scenario, and one of the most representative is the  threshold Paillier-HE \cite{cramer2001multiparty}. For example, \textit{Zheng et al.} \cite{zheng2019helen} propose Helen, the first secure federated training system utilizing the threshold Paillier-HE. In Helen,  each user's data are encrypted with  Paillier-HE and  submitted to an ``Aggregator", which is responsible for performing aggregation.  Then,  the Aggregator broadcasts the  aggregated  results to all the users to update the local model parameters.
 When the trained model reaches the preset convergence condition, the model parameters can be decrypted through the collaboration of multiple users without revealing the original private key.

\textbf{Limitations}:  Threshold Paillier-HE requires substantial modular exponential operations for function evaluation under ciphertext, and requires expensive calculations among multiple users for decryption.  In a decentralized scenario, each user receives gradients from neighboring users and aggregates them to update its local model. This inevitably produces worse performance if threshold Paillier-HE is simply used as its underlying architecture { please refer to Section \ref{sec:PERFORMANCE EVALUATION} for more details}. Other variants like Threshold Fully Homomorphic Encryption (TFHE) \cite{chen2019multi, boneh2018threshold},  are possibly applicable to the  decentralized environment. However,  TFHE can only encrypt one bit at a time, which is obviously unrealistic to achieve practical training.

\textbf{Remark 1}: Based on the above discussions, we argue that differential privacy-based and MPC-based approaches are contrary to our motivation and  the characteristics of decentralized learning systems. In contrast, threshold-based HE seems to be more promising, if it can be freed from the computational bottleneck of ciphertext operation and distributed decryption. Inspired by this, this paper focus on exploring the benefits of a state-of-the-art HE method, Brakerski-Fan-Vercauteren (BFV), and the possibilities to extend it to a multiparty version for privacy-preserving decentralized learning.

\section{Preliminaries}
\label{sec:PROBLEM STATEMENT}
In this section, we first review some basic concepts about decentralized parallel stochastic algorithms and BFV Homomorphic Encryption. Then,  we  describe the threat model and  privacy requirements considered in this paper.

\subsection{Decentralized Parallel Stochastic Algorithms}
\label{sub:Decentralized parallel stochastic algorithms}
As shown in Figure~\ref{Fig:Different types of training frameworks}(c), a decentralized system can be represented as an undirected graph $(V, E)$, where $V$ denotes a set of $N$ nodes in the graph (i.e., users in the system\footnote{In this paper we use the terminologies of node and user interchangeably.}), and $E$ denotes a set of edges representing communication links. We have $(i, j) \in E$ if and only if node $i$ can receive information from node $j$. $\mathcal{N}_i=\{j| (i, j) \in E\}$ represents the set of all nodes connected to node $i$. $E\in \mathbb{R}^{N\times N}$ is a symmetric doubly stochastic matrix to denote the training dependency of two nodes. It has the following two properties: (i) $E_{i,j}\in[0,1]$ and (ii) $\sum_{j}E_{i,j}=1$ for all $i$. Commonly for a node $i$, we can set $E_{i,j}=0$ if node $j\notin \mathcal{N}_i$ and $E_{i,j}=1/|\mathcal{N}_i|$ otherwise.


For neural network training under such a decentralized system, all nodes are required to optimize the following function \cite{vogels2020practical,Koloskova2020Decentralized}:

\begin{small}
 \begin{equation}
\begin{split}
 \min_{W\in\mathbb{R}^{H}}G(W)=\frac{1}{N}\sum_{i=1}^{N}\mathbb{E}_{{X}\sim D_{i}}L_i(W, X)
\end{split}
\end{equation}
\end{small}

where $W\in\mathbb{R}^{H}$ denotes the parameters of the target model. The distribution of training samples for each user is denoted as $D_i$, and $X\in \mathbb{R}^{M}$ represents  a training sample from the distribution.  For each node $i$, $L_i(W,X)=L(W, X)$ denotes the loss function.

 Decentralized parallel stochastic gradient descent algorithm (D-PSGD) \cite{Koloskova2020Decentralized} is usually used to solve the above optimization problem. Its main idea is to update the local model by requiring each user to exchange gradients with their neighbors (technical details are shown in \textbf{Algorithm~\ref{algorithm 1}}). In this paper, our goal is to implement D-PSGD in a privacy-preserving way.

\begin{algorithm}
\begin{small}
\caption{Decentralized parallel stochastic gradient descent}
\label{algorithm 1}
\begin{algorithmic}[1]
\REQUIRE Initialize $W_{0,i}=W_0$, matrix $E$, step length $\eta$, and the number of iterations $K$.
\FOR {$k=0, 1, 2, \cdots K-1$}
\STATE Each node $i$ performs the following operations in parallel:
\STATE Randomly select sample $X_{k,i}$ from the local dataset and calculate local stochastic gradient $\Delta W_i=\triangledown L_i(W_{k,i}, X_{k,i})$. Note that we can use mini-batch of stochastic gradients to accelerate this process without hurting accuracy.
\STATE Obtain the current gradient $W_{k,j}$ of all nodes in $\mathcal{N}_i$, and calculate the weighted average as follows\footnotemark[2]:
\begin{equation*}
\begin{split}
W_{k+\frac{1}{2}, i}=E_{i,i}W_{k,i}+\sum_{j\in \mathcal{N}_i}E_{i,j}W_{k,j}
\end{split}
\end{equation*}
\STATE Update local model parameters $W_{k+1,i}\leftarrow W_{k+\frac{1}{2}, i}-\eta \Delta W_i$
\STATE Broadcast $W_{k+1,i}$ to all nodes in $\mathcal{N}_i$.
\ENDFOR
\ENSURE  $\frac{1}{N}\sum_{i=1}^{N}W_{K, i}$.
\end{algorithmic}
\end{small}
\end{algorithm}
\footnotetext[2]{Lines 3 and 4 can be run in parallel.}
\subsection{BFV Homomorphic Encryption }
\label{sub:Decentralized parallel stochastic algorithms}
The BFV cryptosystem \cite{fan2012somewhat} is a fully homomorphic encryption scheme based on the Ring-learning with error (RLWE)  problem. It supports both addition and multiplication operations in ciphertext. In this section, we briefly introduce the basic principles of the standard BFV algorithm used in the centralized scenario. In Section~\ref{subsub:PROPOSED SCHEME}, we will explain in detail how to convert this standard BFV to the multiparty version and enable decentralized training.

Suppose the ciphertext space is composed of polynomial ring $R_q=\mathbb{Z}_q[\mathtt{X}]/(\mathtt{X}^n+1)$, and the  quotient ring of polynomials with coefficients in $\mathbb{Z}_q$, where $\mathtt{X}^n+1$ is a monic irreducible polynomial with degree of $n=2^b$. The set of integers in $(\frac{-q}{2},\frac{q}{2}]$ is used to denote the representatives for the congruence classes modulo
$q$.  Similarly,   the plaintext space is denoted as the ring $R_t=\mathbb{Z}_t[\mathtt{X}]/(\mathtt{X}^n+1)$ where $t<q$. We use $\Lambda=\left\lfloor q/t\right\rfloor$ to represent the integer division of $q$ by $t$. Unless otherwise stated, we consider the arithmetic under $R_q$.  Therefore, the symbol of polynomial reductions is sometimes omitted in the BFV execution.   Informally, the standard BFV encryption system consists of the following five algorithms.

1. $\mathsf{BFV\cdot SecKeyGen}(1^\lambda)\rightarrow sk$: Given the security parameter $\lambda$, this algorithm selects an element $s$ uniformly on the polynomial ring $R_3=\mathbb{Z}_3[\mathtt{X}]/(\mathtt{X}^n+1)$,  where the coefficients of every polynomial in $R_3$ are uniformly distributed in $\{-1, 0, 1\}$. Then, it outputs the secret key $sk=s$.

2. $\mathsf{BFV\cdot PubKeyGen}(sk)\rightarrow pk$: Given the secret key $s$, this algorithm selects an element $p_1$ uniformly on the polynomial ring $R_q$, and an error term $e$ from  $\chi$.  $\chi$ is a distribution over $R_q$ with coefficients obeying the central discrete Gaussian with standard deviation $\sigma$ and truncated to support  over $[B, B]$. Then, it outputs $pk=(p_0, p_1)=(-(sp_1+e), p_1)$.

3. $\mathsf{BFV\cdot Encrypt}(pk, x)\rightarrow ct$: Given the public key $pk=(p_0, p_1)$, this algorithm samples an element $\mu$ uniformly from $R_3$, and two  error terms $e_0$, $e_1$ from $\sigma$. Then, it outputs the ciphertext $ct=(\Lambda x+\mu p_0+e_0, \mu p_1+e_1)$.

4. $\mathsf{BFV\cdot Eval}$ $(pk, f, ct_1,$$ ct_2, \cdots, ct_N)\rightarrow ct'$: Given the public key $pk$, the function $f$ to be evaluated, and  $N(N\geq 1)$ ciphertext inputs $(ct_1,$$ ct_2, \cdots, ct_N)$, this algorithm outputs the ciphertext  result $ct'$. Note that since BFV supports any numbers of addition and multiplication operations in ciphertext, it is feasible to securely evaluate a function $f$ that can be (approximately) parsed as a polynomial. To achieve this, BFV uses  $\mathsf{BFV\cdot Add}$ and $\mathsf{BFV\cdot Mul}$  operations to perform homomorphic addition and multiplication  respectively, and uses the \textit{relinearization key (rlk)} to ensure the consistency of the ciphertext form after each multiplication. In addition, $\mathsf{BFV\cdot Bootstrap}$ is used to reduce the noise of the ciphertext  back to a fresh-like one, which enables further calculations even if the noise of the current ciphertext reaches the limit of the homomorphic capacity.  Readers can refer to \cite{fan2012somewhat} for more details.

5. $\mathsf{BFV\cdot Decrypt}(sk, ct)\rightarrow x$: Given the secret key $s$ and the ciphertext $ct=(c_0, c_1)$, this algorithm outputs the decrypted plaintext $x=\left[\left\lfloor \frac{t}{q}[c_0+c_1s]_q\right\rceil\right]_t$, where $[c_0+c_1s]_q$ denotes $c_0+c_1s \mod q$.

 The security of the BFV cryptosystem is reduced to the famous \textit{Decisional-RLWE Problem} \cite{cheon2017homomorphic}. Informally, given a random $a$, a secret key $s$ and an error term  $e$ uniformly sampled from $R_q$,  $R_3$ and $\chi$, respectively, it is computationally difficult for an adversary to distinguish the two distributions ($sa+e, a$) and $(g, a)$ without the knowledge of $s$ and $e$, where $g$ is  uniformly sampled from $R_q$.

\subsection{Threat Model and Privacy Requirement }
\label{sub:Thread Model and Privacy Requirements}
As shown in Figure \ref{Fig:Different types of training frameworks}(c), we consider a decentralized learning system with $N$ users. Each user $i$ with a local dataset $D_i$ adopts the D-PSGD algorithm to collaboratively train a deep learning model with others. In this paper, each user is considered to be honest but curious \cite{JuvekarVC18,mohassel2018aby3,Pattuk2016}, i.e., they follow the agreed procedure to perform the training task, but may try to obtain the private data (i.e., gradients) of other users along with the collected prior knowledge. As a result, attacks from malicious adversaries by violating the execution of the protocol are beyond the scope of this paper. Such a threat model has been widely used in existing works about privacy-preserving machine learning \cite{zhanggala,riazi2019xonn,sav2020poseidon,NEURIPS2020_5bf8aaef}. Moreover,  we allow the collusion of the majority of users to enhance the attack capabilities. Specifically, for the union composed of user $i$ and its connected node set $\mathcal{N}_i$, i.e., $\mathcal{U}=i\cup \mathcal{N}_i$, collusion of at most $|\mathcal{U}-1|=|\mathcal{N}_i|$ users is allowed at any training stage.   Our goal is to protect the confidentiality of sensitive data (i.e., gradients) for each benign user. This means during the training process, we should guarantee that no user $i$ can learn the gradient $\Delta W_j$ of  any benign user $j$, except those can be inferred from its own input data $\Delta W_i$.

  \section{PROPOSED SCHEME}
\label{subsub:PROPOSED SCHEME}
We present a novel privacy-preserving framework, Decentralized Deep Learning with Multiparty Homomorphic Encryption (\textbf{D$^{2}$-MHE}), which enables $N$ users to train the target model collaboratively under the decentralized network. We first give the overview of \textbf{D$^{2}$-MHE} for implementing the D-PSGD algorithm with the multiparty version of BFV, and then further explain the detailed algorithms.

 \subsection{Overview}
 \begin{figure*}
  \centering
 \includegraphics[width=1.0\textwidth]{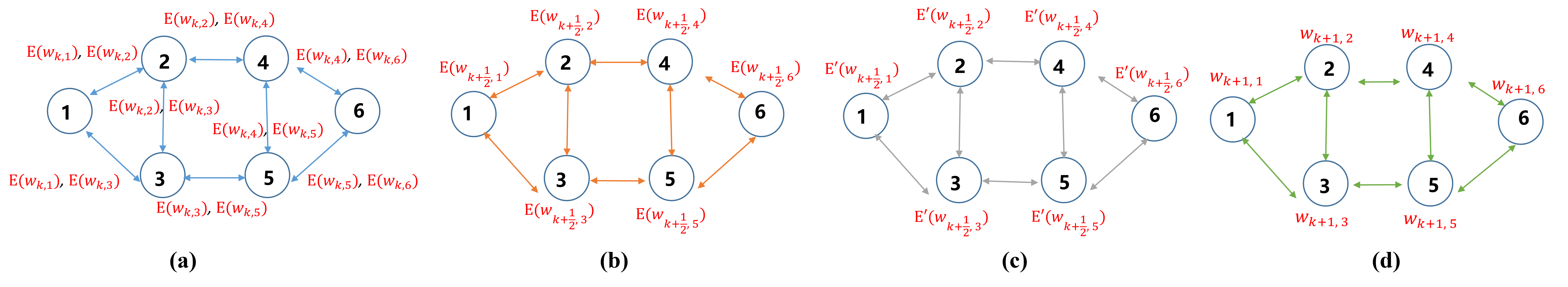}
  \caption{ A high-level view of our \textbf{D$^{2}$-MHE}. (a) Step 1: Exchange encrypted local parameters with connected nodes.(b) Step 2: Each node locally aggregates parameters from other nodes. (c) Step 3: Each node interacts with connected nodes to convert the local aggregated value into ciphertext under its own public key. (d) Step 4: Each node decrypts the aggregated value and updates the local parameters.}
  \label{Fig:A high level view}
\end{figure*}
Essentially, in \textbf{D$^{2}$-MHE}, users iteratively execute the D-PSGD algorithm with the multiparty BFV cryptosystem.  The complete algorithm is shown in \textbf{Algorithm~\ref{algorithm 2}}, where the newly constructed functions (marked in red) are introduced to convert the standard BFV into a multiparty version. During decentralized training, each node requires additional operations to securely execute the D-PSGD compared to the original algorithm, including generating additional variables for encryption/decryption, and modifying certain operations  to implement the ciphertext calculations. We describe each step in Algorithm~\ref{algorithm 2} as well as Figure~\ref{Fig:A high level view}.

\begin{algorithm*}
\begin{small}
\caption{Privacy-preserving D-PSGD}
\label{algorithm 2}
\begin{algorithmic}[1]
\REQUIRE  Each node $i$ independently generates its secret key $sk_i\leftarrow \mathsf{BFV\cdot SecKeyGen}(1^\lambda)$ and public key $pk_i \leftarrow \mathsf{BFV\cdot PubKeyGen}(sk_i)$. For each set $\mathcal{N}_i, i\in [1,N]$, nodes $i$ and $j \in \mathcal{N}_i$ generate shares  \textcolor{red}{$(s_i, s_j|j\in \mathcal{N}_i) \leftarrow \mathsf{MBFV\cdot SecKeyGen}(1^\lambda)$} of the system secret key $sk=s$, where we have $s=s_i+\sum_{j\in\mathcal{N}_i}s_j$. Then,  initialize $W_{0,i}=W_0$, matrix $E$, step length $\eta$, and the number of iterations $K$.
\STATE  Each node $i$ cooperates with the nodes in $\mathcal{N}_i$ to generate the system public key  \textcolor{red} {$pk \leftarrow \mathsf{MBFV\cdot PubKeyGen}(s_i, s_j|j\in \mathcal{N}_i)$}
\STATE Each node $i$ broadcasts its public key $pk_i$ to all nodes $j\in \mathcal{N}_i$.
\FOR{$k=0, 1, 2, \cdots K-1$}
\STATE \COMMENT{Each node $i$ performs the following operations in parallel:}
\STATE Randomly select sample $X_{k,i}$ from the local dataset and calculate local stochastic gradient $\Delta W_i=\triangledown L_i(W_{k,i}, X_{k,i})$
\STATE Obtain the current encrypted gradient $\mathsf{E}(W_{k,j})\leftarrow \mathsf{BFV\cdot Encrypt}(pk, W_{k,j})$ of all nodes $j$ in $\mathcal{N}_i$, and calculate the encrypted weighted average with the help of the algorithm \textcolor{red}{$\mathsf{MBFV\cdot Bootstrap}(ct, s_i, s_j|j\in \mathcal{N}_i)$} as follows:
\begin{equation*}
\begin{split}
\mathsf{E}(W_{k+\frac{1}{2}, i}) \leftarrow \mathsf{BFV\cdot Eval}(pk, f, \mathsf{E}(W_{k,i}), \mathsf{E}(W_{k,j})| j\in \mathcal{N}_i)
\end{split}
\end{equation*}
where $f=E_{i,i}W_{k,i}+\sum_{j\in \mathcal{N}_i}E_{i,j}W_{k,j}$.
\STATE Broadcast $\mathsf{E}(W_{k+\frac{1}{2}, i})$ to all nodes $j\in \mathcal{N}_i$.
\STATE Work with other nodes $j\in \mathcal{N}_i$ to convert  $\mathsf{E}(W_{k+\frac{1}{2}, i})$ into a new ciphertext  \textcolor{red}{$\mathsf{E}'(W_{k+\frac{1}{2}, i})\leftarrow \mathsf{MBFV\cdot Convert}(\mathsf{E}(W_{k+\frac{1}{2}, i}), pk_i, s_i, s_j|j\in \mathcal{N}_i)$} under the public key $pk_i$.
\STATE Decrypt $\mathsf{E}'(W_{k+\frac{1}{2}, i})$ as $W_{k+\frac{1}{2}, i}\leftarrow \mathsf{BFV\cdot Decrypt}(sk_i, \mathsf{E}'(W_{k+\frac{1}{2}, i}))$.
 \STATE Update local model parameters $W_{k+1,i}\leftarrow W_{k+\frac{1}{2}, i}-\eta \Delta W_i$.
\STATE Broadcast $\mathsf{E}(W_{k+1,i})\leftarrow \mathsf{BFV\cdot Encrypt}(pk, W_{k+1,i})$ to all nodes in $\mathcal{N}_i$.
\ENDFOR
\ENSURE  $\frac{1}{N}\sum_{i=1}^{N}W_{K, i}$.
\end{algorithmic}
\end{small}
\end{algorithm*}

(1) In the initialization phase, each node $i$  generates its own secret key $sk_i$ and public key $pk_i$ utilizing the  standard BFV. Then, for each set $\mathcal{N}_i, i\in [1,N]$, a new function $\mathsf{MBFV\cdot SecKeyGen}$  is used to  generate  shares  $(s_i, s_j|j\in \mathcal{N}_i)$ of the system secret key $sk=s$, where  $s=s_i+\sum_{j\in\mathcal{N}_i}s_j$.  Then, a new function $\mathsf{MBFV\cdot PubKeyGen}$ is  exploited to generate the public key $pk$ corresponding to $s$.

(2) For each node $j$, instead of sending the original gradient $W_{k,j}$ to other neighboring nodes, it uses the standard $\mathsf{BFV\cdot Encrypt}$  to send the ciphertext $\mathsf{E}(W_{k,j})$ (Line 6 and Figure~\ref{Fig:A high level view}(a))  to all the connected users $i\in \mathcal{N}_j$.

(3) Each node $i$ computes the encrypted weighted average $\mathsf{E}(W_{k+\frac{1}{2}, i})$ with the help of a new function $\mathsf{MBFV\cdot Bootstrap}$ (Figure~\ref{Fig:A high level view}(b)), which is the multiparty bootstrapping procedure.  $\mathsf{MBFV\cdot Bootstrap}$ can  reduce the noise of a ciphertext (such as the intermediate value $ct$ in Line 6)  back to a fresh-like one, which enables further calculations even if the noise of the current ciphertext reaches the limit of homomorphic capacity.

(4) To  securely decrypt $\mathsf{E}(W_{k+\frac{1}{2}, i})$,  we construct a new function $\mathsf{MBFV\cdot Convert}$, which can  obliviously re-encrypt $\mathsf{E}(W_{k+\frac{1}{2}, i})$ that is originally encrypted under the system public key $pk$, into a new ciphertext $\mathsf{E}'(W_{k+\frac{1}{2}, i})$ under the recipient's public key $pk_i$ (Line 8 and Figure~\ref{Fig:A high level view}(c)).

(5) As a result, each node $i$ can decrypt $\mathsf{E}'(W_{k+\frac{1}{2}, i})$ with its secret key $sk_i$, and then update  its local model parameters (Lines 9-10  and Figure\ref{Fig:A high level view}(d)).

To sum up, in \textbf{D$^{2}$-MHE}, we construct four new functions ($\mathsf{MBFV\cdot SecKeyGen}$, $\mathsf{MBFV \cdot PubKeyGen}$, $\mathsf{MBFV\cdot Bootstrap}$, and $\mathsf{MBFV\cdot Convert}$)  based on the standard BFV cryptosystem, which are used   to support system secret key generation, public key construction, distributed bootstrapping procedure, and ciphertext conversion in a decentralized learning environment. Note that the functions $\mathsf{MBFV\cdot SecKeyGen}$ and $\mathsf{MBFV\cdot PubKeyGen}$ are only executed once during the entire training process. Besides, all nodes encrypt their local gradients under the same public key $pk$ and broadcast to other nodes.  As a result, compared with existing MPC-based works, where each node  needs to secretly share its gradients to all neighbors,  our method only requires each user to broadcast a ciphertext to all users. Therefore, from the perspective of the whole system,  \textbf{D$^{2}$-MHE} reduces the communication overhead of each round of gradient update  from quadratic to linear  without sacrificing the accuracy of the original model.

\textbf{Remark 2}: Our \textbf{D$^{2}$-MHE} is inspired by work \cite{mouchet2020multiparty}, which proposes a  cryptographic primitive called multiparty homomorphic encryption from ring-learning-with-errors. However, there are three major differences between \cite{mouchet2020multiparty} and our work, making \cite{mouchet2020multiparty} incompatible to our scenario. (1) \cite{mouchet2020multiparty} focuses on the construction of relinearization keys for multiparty HE, thereby ensuring the correctness of multiplication between ciphertexts. It's hard to apply it to our system which mainly consists of ciphertext aggregation operations rather than multiplications. (2) \cite{mouchet2020multiparty} considers to convert a ciphertext originally encrypted under the system secret key into a new ciphertext under the recipient's secret key which is securely shared with other users. In our system, we need access to the recipient's public key instead of the secret key. (3) \cite{mouchet2020multiparty} is mainly designed for scenarios such as private information-retrieval and private-set intersection, while our scheme is tailored to decentralized learning. Due to the above differences, compared with \cite{mouchet2020multiparty}, our proposed algorithms  are more concise and efficient, and is more suitable for decentralized deep learning.

 \subsection{Detailed Implementations of the Four Functions}
 \label{Multiparty BFV for D-PSGD}

We provide the details of the above four newly constructed functions in \textbf{D$^{2}$-MHE}. For readability,  we take the variables in \textbf{Algorithm~\ref{algorithm 2}} as inputs/outputs of the functions.

1. $\mathsf{MBFV\cdot SecKeyGen}(1^\lambda) \rightarrow (s_i, s_j|j\in \mathcal{N}_i)$:  Given a security parameter $\lambda$, this function generates  shares $(s_i, s_j|j\in \mathcal{N}_i)$ of the system secret key $sk=s$  for each set $\mathcal{N}_i, i\in [1,N]$. In this paper, we focus on additive secret sharing  \cite{mouchet2020multiparty} of the key, i.e., $s=s_i+\sum_{j\in\mathcal{N}_i}s_j$, which can also be replaced by the Shamir's threshold secret sharing \cite{benhamouda2021local} with less strict requirements.  We propose a simple method to implement $\mathsf{MBFV\cdot SecKeyGen}(1^\lambda)$, i.e., each user independently generates $s_i$ through the standard $\mathsf{BFV\cdot SecKeyGen}(1^\lambda)$.  Then, we can simply set $s=s_i+\sum_{j\in\mathcal{N}_i}s_j$. It means that $s$ is not set beforehand, but is determined when all nodes have generated their own shares.  The advantage of generating $s$ in this way is that each node $i$ does not need to share its $s_i$ with others. It is  a common way to generate a collected key and has been  proven to be secure \cite{cramer2001multiparty, boyle2015function}\footnote{The $s$ generated in this way may not conform to the property of being uniformly distributed under $R_3$. However, this is not a problem because our security proof (refer to Section~\ref{Security Analysis}) does not rely on this property. Also, there are many other ways \cite{zheng2019helen,mohassel2017secureml} to generate uniform keys that are subject to the distribution $R_3$, which require private channels between users.}.

2. $\mathsf{MBFV\cdot PubKeyGen}(s_i, s_j|j\in \mathcal{N}_i) \rightarrow pk $: This process is to emulate the standard $\mathsf{BFV\cdot PubKeyGen}$  procedure, i.e., generating the public key $pk=(p_0, p_1)$ corresponding to $s$. To achieve this, a public polynomial $p_1$, which is uniformly sampled in the distribution $R_q$, should be agreed in advance by all users. Then, each node $i$ and all nodes $j\in \mathcal{N}_i$  independently sample $e_i, e_j$ over the distribution $\chi$, and compute $p_{0,i}=-(p_1s_i+e_i)$, $p_{0,j}=-(p_1s_j+e_j)$, $j\in \mathcal{N}_i$. Next, each node $k\in{i\cup \mathcal{N}_i}$  broadcasts $p_{0,i}$ to other nodes. Hence, each node $k\in{i\cup \mathcal{N}_i}$ can construct the system public key by performing the following operations:

\begin{small}
\begin{equation*}
\begin{split}
pk=([\sum_{k\in{i\cup \mathcal{N}_i}}p_{0,k}]_q, p_1)= ([-(p_1\sum_{k\in{i\cup \mathcal{N}_i}}s_k+\sum_{k\in{i\cup \mathcal{N}_i}}e_k)]_q, p_1)
\end{split}
\end{equation*}
\end{small}
We observe that $pk$ generated in this way has the same form as the public key generated by the standard $\mathsf{BFV\cdot PubKeyGen}$, but with a larger norm of $||s||$ and $||e||$. The growth of norms is linear with $|\mathcal{N}_i|$, hence it is not a concern (proved in \cite{chen2019multi, boneh2018threshold}), even for a large number of  $|\mathcal{N}_i|$ (See discussion below).

3. {$\mathsf{MBFV\cdot Convert}(\mathsf{E}(W_{k+\frac{1}{2}, i}), pk_i, s_i, s_j|j\in \mathcal{N}_i)\rightarrow \mathsf{E}'(W_{k+\frac{1}{2}, i})$}: This function is used to convert  $\mathsf{E}(W_{k+\frac{1}{2}, i})$ into a new ciphertext  $\mathsf{E}'(W_{k+\frac{1}{2}, i})$ under the public key $pk_i$. As a result, node $i$ can decrypt it with its secret key without accessing the system secret key $s$. To achieve this, given  node $i$'s public key $pk_i=({p_{0,i}, p_{1,i}})$, and  $\mathsf{E}(W_{k+\frac{1}{2}, i})=(c_0, c_1)$,  each node $k\in{i\cup \mathcal{N}_i}$ samples $\mu_k$, $e_{0,k}$, $e_{1,k}$ over the distribution $\chi$ and executes the following operations:
\begin{small}
\begin{equation}
\begin{split}
(h_{0,k}, h_{1,k})=(s_kc_1+\mu_k p_{0,i}+e_{0,k}, \mu_k p_{1,i}+e_{1,k})
\end{split}
\end{equation}
\end{small}

Then, each $(h_{0,k}, h_{1,k})$ is submitted to node $i$. Afterwards, node $i$  first computes $h_0=\sum_{k\in{i\cup \mathcal{N}_i}}h_{0,k}$, and $h_1=\sum_{k\in{i\cup \mathcal{N}_i}}h_{1,k}$, and then generates the new ciphertext  $\mathsf{E}'(W_{k+\frac{1}{2}, i})=(c_0', c_1')$=$(c_0+h_0, h_1)$.

The correctness of $\mathsf{MBFV\cdot Convert}$ is shown as follows:  Given  node $i$'s public key $pk_i=({p_{0,i}, p_{1,i}})$, and  $\mathsf{E}(W_{k+\frac{1}{2}, i})=(c_0, c_1)$, where $c_0+sc_1=\Delta m +e_{ct}$, $p_{0,i}=-(sk_ip_{1,i}+e_{cl})$, we have

\begin{small}
\begin{equation}
\begin{split}
&\mathsf{BFV\cdot Decrypt}(sk_i, \mathsf{E}'(W_{k+\frac{1}{2}, i}))\\
&=\lfloor\frac{t}{q}[c_0+\sum_{k\in{i\cup \mathcal{N}_i}}(s_kc_1+\mu_k p_{0,i}+e_{0,k})\\
&+sk_i\sum_{k\in{i\cup \mathcal{N}_i}}(\mu_k p_{1,i}+e_{1,k})]_q \rceil\\
&=\lfloor \frac{t}{q} [c_0+sc_1+\mu p_{0,i}+sk_i \mu p_{1,i}+e_0+sk_ie_1]_q\rceil\\
&=\lfloor \frac{t}{q} [\Lambda W_{k+\frac{1}{2}}+e_{ct}+e_{Covt} ]_q\rceil\\
&=W_{k+\frac{1}{2}}
\end{split}
\end{equation}
\end{small}

where $e_d=\sum_{k\in{i\cup \mathcal{N}_i}}e_{d,k}$ for $d=0,1$. $\mu=\sum_{k\in{i\cup \mathcal{N}_i}}\mu_k$. Therefore, the additional noise involved in $\mathsf{MBFV\cdot Convert}$ is $e_{Covt}=e_0+sk_ie_1+\mu e_{cl}$, which needs to satisfy the condition of $||e_{ct}+e_{Covt}||<q/(2t)$ for the correctness of decryption.

4. $\mathsf{MBFV\cdot Bootstrap}(ct, s_i, s_j|j\in \mathcal{N}_i) \rightarrow ct'$:  This is the multiparty bootstrapping procedure. It  can  reduce the noise of a ciphertext $ct$  back to a fresh-like one $ct'$, and then enables further calculations  if the noise of the current ciphertext reaches the limit of homomorphic capacity.  Specifically, given a ciphertext $ct=(c_0, c_1)$ with noise variance $\sigma_{ct}^{2}$, a common random polynomial $\alpha$,  each node $k\in{i\cup \mathcal{N}_i}$ samples $M_k$ over $R_t$, $e_{0,k}$, $e_{1,k}$ over $\chi$, and executes the following operations:

\begin{small}
\begin{equation}
\begin{split}
(\eta_{0,k}, \eta_{1,k})=(s_kc_1-\Lambda M_k +e_{0,k}, -s_k\alpha+\Lambda M_k+e_{1,k})
\end{split}
\end{equation}
\end{small}
Then, each $(\eta_{0,k}, \eta_{1,k})$ is submitted to node $i$. Afterwards, node $i$  first computes $\eta_0=\sum_{k\in{i\cup \mathcal{N}_i}}\eta_{0,k}$, and $\eta_1=\sum_{k\in{i\cup \mathcal{N}_i}}\eta_{1,k}$, and then generates the new ciphertext  {$ct'$=($\left[\left\lfloor \frac{t}{q}[c_0+\eta_0]_q\right\rceil\right]_t\Lambda+\eta_1, \alpha$) with noise variance $N\sigma^2$. \\ 

\subsection{Discussions}
\noindent\textbf{Noise analysis}. We can observe that the incurred noise is $\sum_{k \in{i\cup \mathcal{N}_i}}e_k$ in $\mathsf{MBFV\cdot PubKeyGen}$, and $e_{Covt}=e_0+sk_ie_1+\mu e_{cl}$ in $\mathsf{MBFV\cdot Convert}$.  All the noise is controllable since we can preset the range of $e$, $sk_i$ and $\mu$ (please refer to literature \cite{fan2012somewhat} for more theoretical analysis). This stems from our carefully constructed noise mechanism. Since the ciphertext has large size and is difficult to be removed during the decryption process,  our criterion is to keep the accumulated noise items without ciphertext.

\noindent\textbf{The utility of $\mathsf{MBFV\cdot Bootstrap}$}. The implementation of $\mathsf{MBFV\cdot Bootstrap}$ requires the interaction between multiple nodes, which will increase the communication overhead of each user.  However,  $\mathsf{MBFV\cdot Bootstrap}$ is rarely used in our scenarios. In detail,   according to the standard BFV, a ciphertext is correctly decrypted if the noise contained in the ciphertext satisfies  $||e_{ct}||<q/2t$, where $q$ and $t$ denote the ciphertext and plaintext spaces, respectively.  In comparison,  the  noise of a ciphertext involved in \textbf{D$^{2}$-MHE}  is  $e_{ct}+e_{Covt}\approx M \times e_{ct}$, where $M$  can be roughly parsed as a linear function of the average number of adjacent nodes of each node in the system. The noise scale will be further increased to $M \times T \times e_{ct}$, if $T$ times of homomorphic addition operations are performed without utilizing the bootstrapping.  Since the size of $||e_{ct}||$ is usually smaller than 1, for the correctness of decryption, it is enough to ensure that $M \times T<q/2t$.  Hence,  given a 64-bit plaintext space and 512-bit ciphertext space,  we only need to guarantee $M \times T<\frac{2^{512}}{2^{64}}=2^{448}$. Therefore, assuming $M=1024$, \textbf{D$^{2}$-MHE} can still perform $2^{438}$ consecutive homomorphic additions without the assistance of bootstrapping.

In summary, $\mathsf{MBFV\cdot Bootstrap}$ provides a trade-off between computation overhead and communication overhead. It is very practical for computing a function without the knowledge of the computation complexity in advance.

\section{Security Analysis}
\label{Security Analysis}
We now discuss the security of \textbf{D$^{2}$-MHE}.  It can be seen from Section~\ref{Multiparty BFV for D-PSGD} that compared with the standard BFV,  \textbf{D$^{2}$-MHE} constructs four new functions: $\mathsf{MBFV\cdot SecKeyGen}$, $\mathsf{MBFV \cdot PubKeyGen}$, $\mathsf{MBFV\cdot Bootstrap}$, and $\mathsf{MBFV\cdot Convert}$.  Since the implementation of $\mathsf{MBFV\cdot SecKeyGen}$ is essentially calling the standard $\mathsf{BFV\cdot SecKeyGen}$ multiple times, it inherits the security of the original algorithm. Therefore, this section focuses on the security of the other three functions. In addition, in \textbf{D$^{2}$-MHE},  each user $i$  interacts with the connected nodes $\mathcal{N}_i$, while being separated from other users in the system. Hence, we take the set $\mathcal{U}=i\cup \mathcal{N}_i$ as the object of discussion. In brief, the security of \textbf{D$^{2}$-MHE} is mainly tied to the \textit{Decisional-RLWE Problem} \cite{cheon2017homomorphic} and the property of Additive Secret-Sharing \cite{mouchet2020multiparty}. Here we provide arguments in a real/ideal simulation formalism \cite{canetti2014practical}.

Before explaining the details of the proof, we define some variables, which are useful for the subsequent descriptions. Specifically, suppose that the security parameter of \textbf{D$^{2}$-MHE} is $\lambda$, the adversary set is $\mathcal{A}\subseteq \mathcal{U}$, and $|\mathcal{A}|\leq |\mathcal{U}|-1$. $\mathsf{REAL}_{\mathcal{U}}^{\mathcal{U},\lambda}$ is a random variable used to refer to the joint view of all users in  $x_{\mathcal{U}}$, which contains all users' input in \textbf{D$^{2}$-MHE} and information received from other users. Since there is at least one honest user in the set $\mathcal{U}$, we define this honest user as $g_h$ for the convenience of description. The set $\mathcal{H}=\mathcal{U}\backslash(\mathcal{A} \cup g_h)$ represents other honest users. With these symbols, the sketch of our proof is that for any adversary set $\mathcal{A}$, when only the input and output of $\mathcal{A}$ are provided, there  exists a simulator  $\mathsf{SIM}$ with Probabilistic Polynomial Time (PPT) computation ability, which can simulate the view of $\mathcal{A}$, and make $\mathcal{A}$  unable to distinguish the real view from the simulated one.

\subsection{Analysis of $\mathsf{MBFV \cdot PubKeyGen}$}
\label{Security Analysis1}
We consider an adversary set $\mathcal{A}$ to  attack  $\mathsf{MBFV \cdot PubKeyGen}$} defined in Section~\ref{Multiparty BFV for D-PSGD}. For each user $k\in \mathcal{U}$, its private inputs are $s_k$, and $e_k$, and the output  received from the the function is the system public key $pk$. Therefore, given  $\mathcal{A}$'s inputs $\{ s_k, e_k\}, k\in \mathcal{A}$ and $pk=(p_0, p_1)$, the simulator needs to construct a simulated view which is indistinguishable from the adversary's view under the implementation of the real protocol.

\begin{theorem}
Given the security parameter $\lambda$, user set $\mathcal{U}$, adversary set $\mathcal{A}\subseteq \mathcal{U}$,$|\mathcal{A}|\leq |\mathcal{U}|-1$, $\mathcal{A}$'s inputs $\{ s_k, e_k\}_{k\in \mathcal{A}}$, $pk=(p_0, p_1)$, honest user $g_h$, and $\mathcal{H}=\mathcal{U}\backslash(\mathcal{A} \cup g_h)$, there exists a PPT simulator  $\mathsf{SIM}$, whose output  is indistinguishable from the real $\mathsf{REAL}_{\mathcal{U}}^{\mathcal{U},\lambda}$ output.
\begin{small}
\begin{equation*}
\begin{split}
& \mathsf{SIM}_{\mathcal{A}}^{\mathcal{U}, \lambda} (\{ s_k, e_k\}_{k\in \mathcal{A}}, pk)\\ \overset{c}{\equiv} & \mathsf{REAL}_{\mathcal{U}}^{\mathcal{U}, \lambda} (\{ s_j, e_j\}_{j\in \mathcal{U}}, pk)
\end{split}
\end{equation*}
\end{small}
\end{theorem}

\begin{proof}
Since $\mathsf{SIM}$ has $\mathcal{A}$'s inputs $\{ s_k, e_k\}, k\in \mathcal{A}$, and the output $pk=(p_0, p_1)$ of $\mathsf{MBFV \cdot PubKeyGen}$, it needs to simulate all $p_{0,j}=[-(p_1s_j+e_j)]_q, j\in \mathcal{U}$  under two constraints: (i) the  sum of all simulated $p_{0,j}$ and those generated  by $\mathcal{A}$ must be equal to $p_0$, and (ii) the simulated $p_{0,j}$ for $\mathcal{A}$  must be equal to real ones, otherwise the adversary can easily distinguish them. We use the symbol $\tilde{p_{0,j}}$ to denote the simulated shares of $p_0$. $\mathsf{SIM}$   can generate $\tilde{p_{0,j}}$ in the following ways:
$$ \tilde{p_{0,j}}=\left\{
\begin{aligned}
&-[(p_1s_j+e_j)]_q: \; \mathsf{if\; user }\; j\in \mathcal{A}\\
& sample\; from\;  R_q: \; \mathsf{if\; user}\; j\in \mathcal{H}\\
&[p_0-\sum_{j\in{\mathcal{A}\cup \mathcal{H}}}\tilde{p_{0,j}}]_q:  \;\mathsf{if\; user}\; j=g_h \\
\end{aligned}
\right.
$$
We explain how the above simulation guarantees the indistinguishability between $(\tilde{p_{0,1}},\tilde{p_{0,2}}, \cdots \tilde{p_{0,|\mathcal{U}|}})$ and $(p_{0,1}, p_{0,2}, \cdots p_{0,|\mathcal{U}|})$.  Specifically, for each user $j \in  \mathcal{A}$,  since $\mathsf{SIM}$ has $\mathcal{A}$'s inputs $\{ s_j, e_j\}$, it can generate the share $p_{0,j}=[-(p_1s_j+e_j)]_q$, which  is exactly the same as the real value. For each user $j \in  \mathcal{H}$,  $\mathsf{SIM}$ simulates  $p_{0,j}$ by sampling an element uniformly in the distribution $R_q$.   The \textit{Decisional-RLWE Problem} \cite{cheon2017homomorphic} ensures that the sampled value is indistinguishable from the real $p_{0,j}$. Besides, the property of Additive Secret-Sharing \cite{mouchet2020multiparty} makes it a negligible probability  to restore $s_j$ and $e_j$ of the honest user, even if the collusion of multiple users.  For user $j=g_h$, we consider the following two cases: (i) When $\mathcal{H}\neq \emptyset$, $\tilde{p_{0,j}}$ is uniformly random on the distribution $R_q$, as $\sum_{j\in{\mathcal{A}\cup \mathcal{H}}}\tilde{p_{0,j}}$ is a random value distributed on $R_q$. As a result,  the same indistinguishability is achieved as described above. (ii) When $\mathcal{H}= \emptyset$, it means that $|\mathcal{U}-1|$ users are adversaries. Since $pk$ is open to all users, adversaries can reconstruct user $g_h$'s  share through $pk$ and their own knowledge. Therefore, $\mathsf{SIM}$ calculates and outputs the real value for the share of $g_h$.
\end{proof}

\subsection{Analysis of $\mathsf{MBFV\cdot Convert}$}
\label{Security Analysis2}
Similar to the above analysis, we consider an adversary set $\mathcal{A}$ to attack the function $\mathsf{MBFV\cdot Convert}$. The goal of $\mathcal{A}$ is to derive  honest users' shares $\{h_{0,k}, h_{1,k}\}_{k\in \mathcal{H}\cup g_h}$.  Hence,   given  $\mathcal{A}$'s inputs $\{ s_k, \mu_k, e_{0,k}, e_{1, k}\}, k\in \mathcal{A}$, public key $pk_i=(p_{0, i}, p_{1,i})$, and the original ciphertext $\mathsf{E}(W_{k+\frac{1}{2}, i})=(c_0, c_1)$,   the simulator needs to construct a simulated view which is indistinguishable from the adversary's view under the implementation of the real protocol.
\begin{theorem}
Given the security parameter $\lambda$, user set $\mathcal{U}$, adversary set $\mathcal{A}\subseteq \mathcal{U}$,$|\mathcal{A}|\leq |\mathcal{U}|-1$, $\mathcal{A}$'s inputs $\{ s_k, \mu_k, e_{0,k}, e_{1, k}\}_{k\in \mathcal{A}}$, $pk_i=(p_{0, i}, p_{1,i})$, original ciphertext $\mathsf{E}(W_{k+\frac{1}{2}, i})=(c_0, c_1)$,  honest user $g_h$, and $\mathcal{H}=\mathcal{U}\backslash(\mathcal{A} \cup g_i)$, there exists a PPT simulator  $\mathsf{SIM}$, whose output  is indistinguishable from the real $\mathsf{REAL}_{\mathcal{U}}^{\mathcal{U},\lambda}$ output.

\begin{small}
\begin{equation*}
\begin{split}
&\mathsf{SIM}_{\mathcal{A}}^{\mathcal{U}, \lambda} (\{ s_k, \mu_k, e_{0,k}, e_{1, k}\}_{k\in \mathcal{A}}) \\
\overset{c}{\equiv} & \mathsf{REAL}_{\mathcal{U}}^{\mathcal{U}, \lambda} (\{ s_j, \mu_j, e_{0,j}, e_{1, j}\}_{j\in \mathcal{A}})
\end{split}
\end{equation*}
\end{small}
\end{theorem}

\begin{proof}
 We know that $\mathcal{A}$'s inputs  $\{ s_k, \mu_k, e_{0,k}, e_{1, k}\}_{k\in \mathcal{A}}$ are accessible to the $\mathsf{SIM}$. Based on this,  $\mathsf{SIM}$ is required to simulate all $\{h_{0,j}, h_{1,j}\}_{j\in \mathcal{U}}$ under two constraints:
(i) the  sum of all simulated $h_{0,j}$ and $ h_{1,j}$ must be equal to $h_0$ and $h_1$, if the recipient (i.e., user $i$) of the converted ciphertext is malicious, and (ii) the simulated $\{h_{0,k}, h_{1,k}\}$ for $\mathcal{A}$  must be equal to real ones. Otherwise the adversary can easily distinguish them. We use the symbol $\tilde{h_{0,j}}$ and $\tilde{h_{1,j}}$  to denote the simulated shares. $\mathsf{SIM}$   can generate $\tilde{h_{0,j}}$ and $\tilde{h_{1,j}}$  in the following ways:

\begin{small}
$$ (\tilde{h_{0,j}}, \tilde{h_{1,j}})=\left\{
\begin{aligned}
&([s_jc_1+\mu_j p_{0,i}+e_{0,j}]_q,[\mu_k p_{1,i}+e_{1,j}]_q),\\
&\;\;\;\;\;\;\;\;\;\;\;\;\;\;\;\;\;\;\;\;\;\;\;\;\;\;\;\;\;\;\;\;\;\;\;\;\mathsf{if\; user }\; j\in \mathcal{A}.\\
&sample\; from\;  R_q,\mathsf{if\; user}\; j\in \mathcal{H}.\\
&sample\; from\;  R_q,\mathsf{if\; user}\; j=g_h \&\& \mathsf{user}\; i\notin \mathcal{A}\\
& ([h_0-\sum_{j\in{\mathcal{A}\cup \mathcal{H}}}\tilde{h_{0,j}}]_q, [h_1-\sum_{j\in{\mathcal{A}\cup \mathcal{H}}}\tilde{h_{1,j}}]_q),\\
&\;\;\;\;\;\;\;\;\;\;\;\;\;\;\;\;\;\;\;\;\;\;\;\;\;\;\;\mathsf{if\; user}\; j=g_h \&\& \mathsf{user}\; i\notin \mathcal{A}.\\
\end{aligned}
\right.
$$
\end{small}
We explain how the above simulation guarantees the indistinguishability between $(\tilde{h_{0,j}}, \tilde{h_{1,j}})_{j\in \mathcal{U}}$ and $\{h_{0,j}, h_{1,j}\}_{j\in \mathcal{U}}$.  Specifically, for each user $j \in  \mathcal{A}$,  since $\mathsf{SIM}$ has $\mathcal{A}$'s inputs $\{ s_k, \mu_k, e_{0,k}, e_{1, k}\}$, it can generate the share $p(\tilde{h_{0,j}}, \tilde{h_{1,j}})=([s_jc_1+\mu_j p_{0,i}+e_{0,j}]_q, [\mu_k p_{1,i}+e_{1,j}]_q)$, which  is exactly the same as the real value. For each user $j \in  \mathcal{H}$,  $\mathsf{SIM}$ simulates $\{h_{0,j}, h_{1,j}\}$ by sampling an element uniformly in the distribution $\chi$.   The \textit{Decisional-RLWE Problem} \cite{cheon2017homomorphic} ensures that the sampled value is indistinguishable from the real $\{h_{0,j}, h_{1,j}\}$ . Besides, the property of Additive Secret-Sharing \cite{mouchet2020multiparty} makes it a negligible probability  to restore $\{ s_j, \mu_j, e_{0,j}, e_{1, j}\}$ of the honest user, even under the collusion of $|\mathcal{U}|-2$ users.  For user $j=g_h$, we consider the following two cases: (i) When user $i\notin \mathcal{A}$, $\tilde{h_{0,j}}, \tilde{h_{1,j}})$ is uniformly random on the distribution $\chi$.  This is because the adversary cannot access the final values $h_0$ and $h_1$.  So in this case it is not necessary to ensure that the sum of all simulated $h_{0,j}$ and $ h_{1,j}$ is  equal to $h_0$ and $h_1$. Besides,  the same indistinguishability is achieved as described above. (ii) When $i\in \mathcal{A}$, it means that   $h_0$ and $h_1$ are submitted to $\mathcal{A}$.  Hence, $(\tilde{h_{0,j}}, \tilde{h_{1,j}})$ can be constructed as $([h_0-\sum_{j\in{\mathcal{A}\cup \mathcal{H}}}\tilde{h_{0,j}}]_q, [h_1-\sum_{j\in{\mathcal{A}\cup \mathcal{H}}}\tilde{h_{1,j}}]_q)$. As a result, the  sum of all simulated $h_{0,j}$ and $ h_{1,j}$ is equal to $h_0$ and $h_1$. This guarantees the indistinguishability between the simulated view and the real view.
\end{proof}

\subsection{Analysis of $\mathsf{MBFV\cdot Bootstrap}$}
\label{Security Analysis3}

We finally discuss the security of the function $\mathsf{MBFV\cdot Bootstrap}$. Specifically, given the  adversary set $\mathcal{A}$, the goal of $\mathcal{A}$ is to derive  honest users' shares $\{\eta_{0,k}, \eta_{1,k}\}_{k\in \mathcal{H}\cup g_h}$.  The simulator needs to construct a simulated view which is indistinguishable from the adversary's view  with $\mathcal{A}$'s inputs $\{ s_k, M_k, e_{0,k}, e_{1,k}\}, k\in \mathcal{A}$ and the public ciphertext $ct=(c_0, c_1)$.
\begin{theorem}
Given the security parameter $\lambda$, user set $\mathcal{U}$, adversary set $\mathcal{A}\subseteq \mathcal{U}$,$|\mathcal{A}|\leq |\mathcal{U}|-1$, $\mathcal{A}$'s inputs $\{ s_k, M_k, e_{0,k}, e_{1,k}\}_{k\in \mathcal{A}}$,  original ciphertext $ct=(c_0, c_1)$,  honest user $g_h$, and $\mathcal{H}=\mathcal{U}\backslash(\mathcal{A} \cup g_i)$, there exists a PPT simulator  $\mathsf{SIM}$, whose output  is indistinguishable from the real $\mathsf{REAL}_{\mathcal{U}}^{\mathcal{U},\lambda}$ output.

\begin{small}
\begin{equation}
\begin{split}
&\mathsf{SIM}_{\mathcal{A}}^{\mathcal{U}, \lambda} (\{ s_k, M_k, e_{0,k}, e_{1,k}\}_{k\in \mathcal{A}})\\
\overset{c}{\equiv} & \mathsf{REAL}_{\mathcal{U}}^{\mathcal{U}, \lambda} (\{ s_j, M_j, e_{0,j}, e_{1,j}\}_{j\in \mathcal{A}})\\
\end{split}
\end{equation}
\end{small}
\end{theorem}

\begin{proof}
 Given $\mathcal{A}$'s inputs  $\{ s_k, M_k, e_{0,k}, e_{1,k}\}_{k\in \mathcal{A}}$,  $\mathsf{SIM}$  is required to simulate all $\{ \eta_{0,j}, \eta_{1,j}\}_{j\in \mathcal{U}}$ under two constraints:
(i) the  sum of all simulated $\eta_{0,j}$ and $ \eta_{1,j}$ must be equal to $\eta_0$ and $\eta_1$, if the recipient (i.e., user $i$) of the new ciphertext $ct'$ is malicious, and (ii) the simulated $\{\eta_{0,j}, \eta_{1,j}\}$ for $\mathcal{A}$  must be equal to real ones. Otherwise the adversary can easily distinguish them. We use the symbol $\tilde{\eta_{0,j}}$ and $\tilde{\eta_{1,j}}$  to denote the simulated shares. $\mathsf{SIM}$   can generate $\tilde{\eta_{0,j}}$ and $\tilde{\eta_{1,j}}$  in the following ways:

\begin{small}
$$ (\tilde{\eta_{0,j}}, \tilde{\eta_{1,j}})=\left\{
\begin{aligned}
&([s_jc_1-\Lambda M_j +e_{0,j}]_q, [-s_j\alpha+\Lambda M_j+e_{1,j}]_q),\\
& \;\;\;\;\;\;\;\;\;\;\;\;\;\;\;\;\;\;\;\;\;\;\;\;\;\;\;\;\;\;\;\;\;\;\;\;\mathsf{if\; user }\; j\in \mathcal{A}.\\
&sample\; from\, R_q, \mathsf{if\; user}\; j\in \mathcal{H}\\
&sample\; from\;  R_q, \; \mathsf{if\; user}\; j=g_h \&\& \mathsf{user}\; i\notin \mathcal{A}.\\
&([\eta_0-\sum_{j\in{\mathcal{A}\cup \mathcal{H}}}\tilde{\eta_{0,j}}]_q, [\eta_1-\sum_{j\in{\mathcal{A}\cup \mathcal{H}}}\tilde{\eta_{1,j}}]_q),\\
 &\;\;\;\;\;\;\;\;\;\;\;\;\;\;\;\;\;\;\;\;\;\;\;\;\;\;\;\;\mathsf{if\; user}\; j=g_h \&\& \mathsf{user}\; i\notin \mathcal{A}.\\
\end{aligned}
\right.
$$
\end{small}
We explain how the above simulation guarantees the indistinguishability between $(\tilde{\eta_{0,j}}, \tilde{\eta_{1,j}})_{j\in \mathcal{U}}$ and $\{\eta_{0,j}, \eta_{1,j}\}_{j\in \mathcal{U}}$.  Specifically, for each user $j \in  \mathcal{A}$,  since $\mathsf{SIM}$ has $\mathcal{A}$'s inputs $\{ s_j, M_j, e_{0,j}, e_{1,j}\}$, it can generate the share $p(\tilde{\eta_{0,j}}, \tilde{\eta_{1,j}})=([s_jc_1-\Lambda M_j +e_{0,j}]_q, [-s_j\alpha+\Lambda M_j+e_{1,j}]_q)$, which  is exactly the same as the real value. For each user $j \in  \mathcal{H}$,  $\mathsf{SIM}$ simulates $\{\eta_{0,j}, \eta_{1,j}\}$ by sampling an element uniformly in the distribution $R_q$.   The \textit{Decisional-RLWE Problem} \cite{cheon2017homomorphic} ensures that the sampled value is indistinguishable from the real $\{\eta_{0,j}, \eta_{1,j}\}$ . Besides, the property of Additive Secret-Sharing \cite{mouchet2020multiparty} makes it a negligible probability  to restore $\{ s_j, M_j, e_{0,j}, e_{1,j}\}$ of the honest user, even under the collusion of $|\mathcal{U}|-2$ users.  For user $j=g_h$, we consider the following two cases: (i) When user $i\notin \mathcal{A}$, $\tilde{\eta_{0,j}}, \tilde{\eta_{1,j}})$ is uniformly random on the distribution $R_q$.  This is because the adversary cannot access the final values $\eta_0$ and $\eta_1$.  So in this case it is not necessary to ensure that the sum of all simulated $\eta_{0,j}$ and $ \eta_{1,j}$ is  equal to $\eta_0$ and $\eta_1$. Besides,  the same indistinguishability is achieved as described above. (ii) When $i\in \mathcal{A}$, it means that   $\eta_0$ and $\eta_1$ are submitted to $\mathcal{A}$.  Hence, $(\tilde{\eta_{0,j}}, \tilde{\eta_{1,j}})$ can be constructed as $([\eta_0-\sum_{j\in{\mathcal{A}\cup \mathcal{H}}}\tilde{\eta_{0,j}}]_q, [\eta_1-\sum_{j\in{\mathcal{A}\cup \mathcal{H}}}\tilde{\eta_{1,j}}]_q)$. As a result, the  sum of all simulated $\eta_{0,j}$ and $ \eta_{1,j}$ is equal to $\eta_0$ and $\eta_1$. This guarantees the indistinguishability between the simulated view and the real view.
\end{proof}

\section{Performance Evaluation}
\label{sec:PERFORMANCE EVALUATION}
We evaluate the performance of \textbf{D$^{2}$-MHE} in terms of classification accuracy, computation and communication  overheads. Specifically,  we simulate a decentralized system with varied numbers of users using Pytorch, where we make use of Onet\footnote{\url{https://github.com/dedis/cothority}} to build the decentralized communication protocol. The average connection rate is $A_\mathcal{N}=0.2$ (i.e., each user is randomly connected to 20\% of all users in the system). Our multiparty version of BFV is modified based on the standard BFV in the SEAL library \cite{chen2017simple}, where the security parameters are taken as $2048$ and $4096$ respectively to test the performance.  The Smart-Vercauteren ciphertext packing technique \cite{xu2020secure} is used to accelerate the efficiency of encryption and ciphertext computation: we set the plaintext slot to 1024, which can pack 1024 plaintexts into one ciphertext at a time, and support  Single-Instruction-Multiple-Data (SIMD) operations.  We consider two image classification tasks trained in the decentralized learning system\footnote{We choose to train our classification tasks with relatively simple networks for simplicity. Note that one property of the HE-based algorithm is that it does not change the accuracy of the original model. Therefore, even if a deeper neural network is trained under ciphertext, the accuracy obtained is almost the same as the one in plaintext. In other words, the training accuracy and the design of the HE scheme are essentially independent of each other.}: a MLP model with two fully-connected layers (100 and 10 neurons, respectively) for MNIST, and a CNN model with two convolutional layers (kernel size of $3\times 1$ per layer) and three fully-connected layers (384 neurons per hidden layer and 10 neurons in the output layer) for CIFAR10.
All the experiments are conducted on a server running the Centos7.4 OS, equipped with 256GB RAM, 64 CPUs (Intel(R) Xeon(R) Gold 6130 CPU \@ 2.10GHZ) and 8 GPUs (Tesla V100 32G).

We select the following baselines for comparisons. (1) D-PSGD \cite{LianZZHZL17} is the original D-PSGD algorithm without any privacy protection. (2) LEASGD \cite{cheng2018leasgd}, $A(DP)^{2}SGD$ \cite{xu2020dp} and DLDP \cite{cheng2019towards} are the three most advanced decentralized learning algorithms with differential privacy. We reproduce these algorithms using exactly the same parameter configuration in their papers. (3) Threshold Paillier-HE \cite{hazay2019efficient,cramer2001multiparty} is a classic homomorphic encryption algorithm that supports distributed key encryption and decryption operations.  We extend Threshold Paillier-HE to the decentralized mode for comparison. (4) COPML \cite{dawson1994breadth} is a privacy-preserving distributed learning framework based on Shamir's secret sharing protocol, which can be regarded as a special kind of decentralized learning with a connection rate of $A_\mathcal{N}=1$. It is also feasible to adapt this framework to a generalized decentralized network.
\subsection{Classification Accuracy}
\label{sec:Accuracy}
\begin{table*}[!htbp]
\centering
\caption{Classification accuracy of different privacy-preserving approaches}
\label{Classification compared with existing approaches}
\begin{tabular}{c|c|c|c|c|c|c|c}
\Xhline{1pt}
\# of users & Dataset & D-PSGD & LEASGD & $A(DP)^{2}SGD$ & DLDP & \textbf{D$^{2}$-MHE}& Paillier-HE\\
\Xhline{1pt}
\multirow{2}{*}{50}&{MNIST}&91.23\%&87.67\% ($\epsilon=8.71$)&84.68\% ($\epsilon=9.43$)&86.7\% ($\epsilon=9.21)$ &91.17\% & 91.19\%\\
&{CIFAR-10}&65.21\%&57.61\% ($\epsilon=11.23$)&53.44\% ($\epsilon=12.13$) &50.7\% ($\epsilon=9.81$) & 65.13\% &  65.10\%\\
\Xhline{1pt}
\multirow{2}{*}{100}&{MNIST}&92.43\%&88.89\% ($\epsilon=8.91$)&85.82\% ($\epsilon=9.79$) &87.15\% ($\epsilon=9.98$) & 91.18\% & 91.14\% \\
&{CIFAR-10}&68.32\%&59.69\% ($\epsilon=12.49$) & 54.14\% ($\epsilon=13.58$) & 52.9\% ($\epsilon=11.31$) &68.25\% & 68.26\%\\
\Xhline{1pt}
\end{tabular}
\end{table*}
We first discuss the performance of \textbf{D$^{2}$-MHE} on the model classification accuracy.
Table~\ref{Classification compared with existing approaches} shows the performance comparisons of \textbf{D$^{2}$-MHE} with existing approaches in the decentralized settings of 50 and 100 users.  
Compared with D-PSGD, we observe that the accuracy drop of HE-based solutions (including \textbf{D$^{2}$-MHE} and Paillier-HE) is negligible, which is mainly attributed to the losslessness of the HE encryption algorithm. Although HE can only handle integers in ciphertext, existing optimization methods (e.g., fixed-point arithmetic circuit conversion \cite{JuvekarVC18}) ensure that the error of ciphertext evaluation for any floating-point number is maintained within $2^{-d}$ (usually $d \geq 13$).

In contrast, other three works based on differential privacy inevitably result in big accuracy drop even if the privacy budget $\epsilon>8$, which is already vulnerable to various types of privacy inference attacks\footnote{According to \cite{jayaraman2019evaluating}, differential privacy with $\epsilon>1$ will loss its effectiveness for deep learning training.}. Here we take the membership attack \cite{rahman2018membership} under $\epsilon=8$ as an example. Based on the definition of differential privacy \cite{fc4ddc15}, the condition $Pr[F(D)\in \mathbf{S}]\leq e^{8}\times Pr[F(D')\in \mathbf{S}]$ should be guaranteed for any two neighboring sets $D$ and $D'$. In other words, even if the target record detected on the dataset $D$ has a probability of 0.0001, it can be detected with a probability of as high as 0.9999 on  $D'$ containing the record. This allows the adversary to infer the presence or absence of the target record from the training data with very high confidence.

\textbf{Remark 3}:  Note that HE-based schemes always exhibit superiority in accuracy over DP-based schemes, as the latter obtains a proper trade-off between accuracy-privacy by introducing noise. However, the comparison with the DP-based solution is not only to illustrate the advantages of our method in accuracy, a more noteworthy conclusion is that it is still unclear whether DP based algorithms can provide satisfactory accuracy and privacy trade-offs in practical applications. Our experimental results are consistent with the results in  work \cite{jayaraman2019evaluating}, i.e., current mechanisms for differentially private deep learning may rarely offer acceptable accuracy-privacy trade-offs for complex learning tasks. Therefore, one of the main motivations for comparison with DP is to explain the choice to use the HE primitives, which may bring a better accuracy and privacy  performance.

\subsection{Computation Overhead}
\label{sec:Computation Overhead}
\begin{table*}[htbp]
\centering
\caption{Computation overhead of different privacy-preserving approaches for each user (unit: seconds)}
\label{Computation Overhead  existing approaches}
\begin{tabular}{c|c|c|c|c|c|c|c}
\Xhline{1pt}
Key Size & Dataset & Method & Initialization&Encryption& \makecell[c]{Ciphertext\\Evaluation} & Decryption &Total time\\
\Xhline{1pt}
\multirow{6}{*}{2048}&\multirow{3}{*}{MNIST}&D-PSGD &-&-&-&-& 2.306\\
&&\textbf{D$^{2}$-MHE}&0.66&9.96&0.54 &15.04&25.54\\
&& Paillier-HE&2.93&21.93&1.27 &33.41&56.61\\
\cline{2-8}
&\multirow{3}{*}{CIFAR-10}&D-PSGD &-&-&- &-&17.18\\
&&\textbf{D$^{2}$-MHE}&0.67&23.98&1.31&36.23&62.19\\
&& Paillier-HE&3.01&43.69&2.48 &57.24&103.40\\
\Xhline{1pt}
\multirow{6}{*}{4096}&\multirow{3}{*}{MNIST}&D-PSGD&-&-&- &-&2.30\\
&&\textbf{D$^{2}$-MHE}&1.42&21.41&1.24 &32.62&55.27\\
&& Paillier-HE&18.98&139.25&2.38 &148.98&290.61\\
\cline{2-8}
&\multirow{3}{*}{CIFAR-10}&D-PSGD&-&-&-&-&17.23\\
&&\textbf{D$^{2}$-MHE}&1.41&51.56&2.99&78.55&133.10\\
&& Paillier-HE &19.24&335.28&5.75 &358.68&726.71\\
\Xhline{1pt}
\end{tabular}
\end{table*}
We further analyze the computation cost of \textbf{D$^{2}$-MHE}.  In brief, the computation load of each user  mainly depends on the key size used for encryption and the average number of users connected to it.  Intuitively, a user needs more computing resources to handle operations with a larger key size, and interact with more users during the training process. To demonstrate this, we first fix the number of users as 100 in the system, and record the running time of each user in a single iteration (i.e., a gradient update with a mini-batch of 256).   Table~\ref{Computation Overhead  existing approaches} shows the experimental results compared with some baseline methods. To facilitate the analysis, we divide the total computation cost into 4 components: (1) \emph{Initialization} is to prepare the public and secret keys of the system. This only needs to be executed once for both \textbf{D$^{2}$-MHE} and Threshold Paillier-HE. (2) \emph{Encryption} is to encrypt gradients of each user. (3) \emph{Ciphertext Evaluation} is to perform cipertext computation. (4) \emph{Decryption} is conducted to decrypt the final results.

From Table~\ref{Computation Overhead  existing approaches}, we observe that the overhead of \textbf{D$^{2}$-MHE} and Threshold Paillier-HE is larger than D-PSGD, because all gradients are encrypted and processed under ciphertext. However, the overhead of \textbf{D$^{2}$-MHE} is significantly lower than that of Threshold Paillier-HE, especially for large key sizes. This is mainly due to the following two reasons: (i) the key sharing and reconstruction processes in Threshold Paillier-HE (including the Initialization and Decryption phases) are highly affected by the key size. A large key size makes it inevitable to perform modular exponential calculations in a large ciphertext space, thereby completing key distribution and distributed decryption. On the contrary, \textbf{D$^{2}$-MHE} only involves vector operations in the traditional sense, which is much less affected by the key size. (ii) Compared with Threshold Paillier-HE, the BFV cryptosystem is more suitable for the SIMD technology, which can process multiple ciphertexts in parallel more efficiently.

\begin{figure*}[htb]
  \centering
  \subfigure[MNIST]{\includegraphics[width=0.49\textwidth]{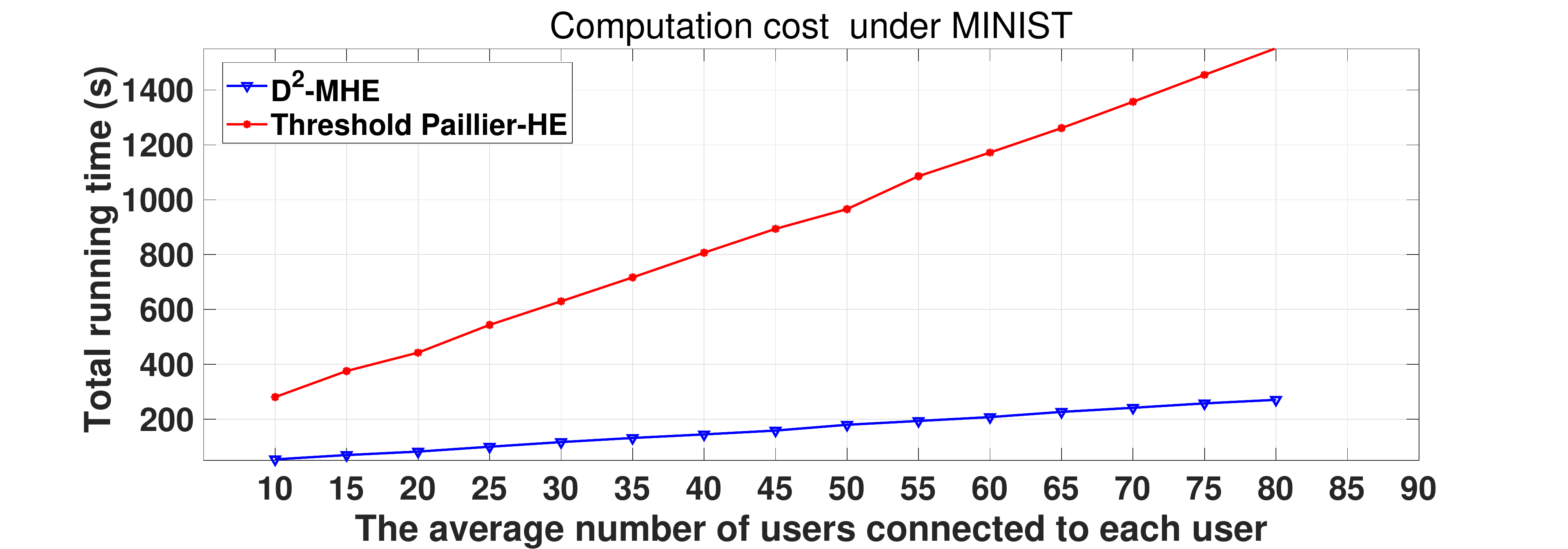}}
 \subfigure[CIFAR10]{\includegraphics[width=0.49\textwidth]{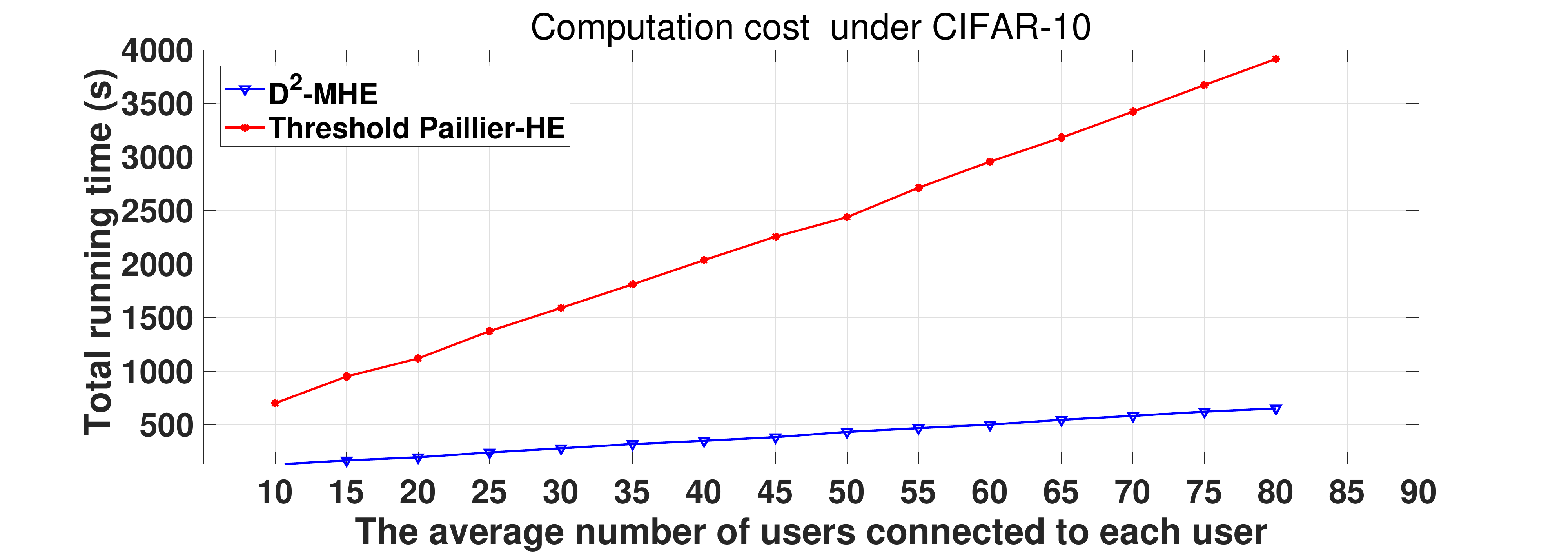}}
  \caption{Total running time of each user for different connection rates.}
  \label{Fig:Computation overhead}
\end{figure*}

We also evaluate the impact of connection rates on the computation overhead of \textbf{D$^{2}$-MHE}. We fix the number of users in the system to 100, and change the connection rate from 0.1 to 0.8. The key size of both  \textbf{D$^{2}$-MHE} and Threshold Paillier-HE is 4096 bit. Smart-Vercauteren ciphertext packing techniques \cite{xu2020secure} is used to accelerate the efficiency of encryption and ciphertext computation, where we set the plaintext slot as 1024 to pack 1024 plaintexts into one ciphertext at a time to support  Single-Instruction-Multiple-Data (SIMD) operations. Figure~\ref{Fig:Computation overhead} shows the running time of each user in a single iteration (i.e., a gradient update with a mini-batch of 256),
We can observe that as the connection rate increases, \textbf{D$^{2}$-MHE} has more significant advantages over Threshold Paillier-HE in terms of computation overhead. This is mainly due to the inefficiency of Threshold Paillier-HE  distributed decryption. As the average number of users connected to each user increases, the number of modular exponential operations performed by Threshold Paillier-HE  increases linearly. As a result, it is quite time-consuming to recover the system secret key under the ciphertext through exponential operations, thereby decrypting the target ciphertext. On the contrary, \textbf{D$^{2}$-MHE} only involves vector operations in the traditional sense, which is much less affected by changes in the connection rate compared to Threshold Paillier-HE.

\subsection{Communication Overhead}
\label{sec:Communication Overhead}

\begin{table*}[!htbp]
\centering
\caption{Theoretical communication overhead of user $i$ for different approaches}
\label{Communication with existing approaches}
\begin{tabular}{c|c|c|c}
\Xhline{1pt}
Method&Gradients-sharing&Aggregation& Total\\
\Xhline{1pt}
D-PSGD&$O(\lambda_3)$&$O(|\mathcal{N}_i|\lambda_3)$&$O((|\mathcal{N}_i|+1)\lambda_3)$\\
COPML&$O(\sum_{j\in N, E_{i,j}\neq 0}|\mathcal{N}_j|\lambda_1)$&$O(|\mathcal{N}_i|\lambda_1)$&$O(\sum_{j\in N, E_{i,j}\neq 0}|\mathcal{N}_j|\lambda_1+|\mathcal{N}_i|\lambda_1)$\\
\textbf{D$^{2}$-MHE}&$O(\lambda_2)$&$O(2|\mathcal{N}_i|\lambda_2)$&$O((2|\mathcal{N}_i|+1)\lambda_2)$\\
\Xhline{1pt}
\end{tabular}
\end{table*}

\begin{table}[!]
\begin{footnotesize}
\centering
 \caption{Experimental communication overhead of user $i$ for different approaches and datasets (unit: MB)}
 \label{Communication with existing approaches}
\begin{tabular}{c|c|c|c|c}
 \Xhline{1pt}
 {Dataset}&Model&Gradients-sharing&Aggregation& Total\\
 \Xhline{1pt}
 \multirow{2}{*}{MNIST}&COPML&$674$&$33.7$&$707.7$\\
 &\textbf{D$^{2}$-MHE}&$18.9$&$37.8$&$56.7$\\
\Xhline{1pt}
\multirow{2}{*}{CIFAR-10}&COPML&$3235$&$161.76$&$3396.76$\\
 &\textbf{D$^{2}$-MHE}&$90.72$&$181.44$&$272.16$\\
 \Xhline{1pt}
 \end{tabular}
 \end{footnotesize}
 \end{table}
We finally analyze the performance of \textbf{D$^{2}$-MHE} in terms of communication overhead,   We theoretically compare the communication complexity of \textbf{D$^{2}$-MHE} with D-PSGD and COPML. The results are shown in Table~\ref{Communication with existing approaches}, where $\lambda_i$ $(i=1, 2, 3)$ denotes the size of a single message. The total computation costs are divided into 2 components (i.e., Gradients-sharing and Aggregation) to facilitate our analysis.
Specifically, in the Gradients-sharing phase, each user $i$ in COPML is required to share its every gradient to each user set $(\mathcal{N}_j|j\in N, E_{i, j}\neq 0)$, which results in the communication complexity of $O(\sum_{j\in N, E_{i,j}\neq 0}|\mathcal{N}_j|\lambda_1)$. In contrast, in \textbf{D$^{2}$-MHE}, the gradients of all users are encrypted with the same public key. As a result,  user $i$ only needs to broadcast a single gradient to other users. In the Aggregation phase (Lines 6-8 in \textbf{Algorithm~\ref{algorithm 2}}), the complexity of COPML is consistent with that of D-PSGD, i.e., receiving information returned by each user to generate the aggregated gradient.  In general, compared to D-PSGD, the communication complexity of \textbf{D$^{2}$-MHE} is only increased by a constant multiple, while the complexity of COPML can reach $O(\sum_{j\in N, E_{i,j}\neq 0}|\mathcal{N}_j|\lambda_1+|\mathcal{N}_i|\lambda_1)$.

It is worth noting that  COPML uses a packed secret sharing method to reduce the communication complexity from $O(\sum_{j\in N, E_{i,j}\neq 0}|\mathcal{N}_j|\lambda_1+|\mathcal{N}_i|\lambda_1)$ to $O(\frac{\sum_{j\in N, E_{i,j}\neq 0}|\mathcal{N}_j|\lambda_1+|\mathcal{N}_i|\lambda_1)}{K})$, where $K$ is the number of secrets packed each time. However, packed secret sharing \cite{cramer2015secure} is restricted to $K<\min(\mathcal{N}_j|j\in N, E_{i, j}\neq 0)$ and only tolerates the  collusion of $\min|(\mathcal{N}_j|j\in N, E_{i, j}\neq 0)|-K$ users at most. On the contrary, we use the Smart-Vercauteren ciphertext packing technique \cite{xu2020secure} to pack multiple plaintexts into one ciphertext, where the number of plaintext slots is independent of the number of users in the system.  As a result,  compared with existing works, \textbf{D$^{2}$-MHE}  has a significant advantage in communication overhead.

We also present the experimental results in terms of communication overhead. We define  the sizes of a single message in COPML and \textbf{D$^{2}$-MHE} as $64$ bit and $4096$ bit, respectively. Such parameters are commonly used to ensure the security of the MPC protocol and HE. In addition, the system has 100 users with a connection rate of 0.2\footnote{ In terms of communication overhead, the connection rate exhibits a linear relationship with each user in COPML, but has no effect on our method. To be precise, the increase of the connection rate makes the number of adjacent nodes of each user increase linearly. Since the communication cost of each user has a positive linear relationship with the number of adjacent nodes, this implies a linear relationship between the connection rate and the communication cost of each user. However, our method only requires each user to broadcast a ciphertext to all users, regardless of the value of the connection rate.}. We iteratively execute  the above two schemes 500 times and 1000 times under the MNIST  and CIFAR-10 datasets respectively.  Then, we record the total communication overhead in Table~\ref{Communication with existing approaches}. For simplicity,  we assume $\min(\mathcal{N}_j|j\in N, E_{i, j}\neq 0)=10$,  thereby the maximum number of secrets shared by the package sharing protocol in COPML is $K<10=9$. Besides, the average number of $(\sum_{j\in N, E_{i,j}\neq 0}|\mathcal{N}_j|)$ is set to 20.  In our \textbf{D$^{2}$-MHE}, the plaintext slot of the Smart-Vercauteren ciphertext packing technique \cite{xu2020secure} is 1024, which can pack 1024 plaintexts into one ciphertext at a time.  We can observe that compared with  COPML, \textbf{D$^{2}$-MHE}  has a significant advantage in the communication overhead. This is mainly due to the large number of interactions in the gradient sharing process of COPML. Moreover,  we use the Smart-Vercauteren ciphertext packing techniques \cite{xu2020secure} to pack multiple plaintexts into one ciphertext, where the number of plaintext slots is independent of the number of users in the system.

\section{Conclusion}
\label{sec:conclusion}
In this work, we propose \textbf{D$^{2}$-MHE}, a practical,  privacy-preserving, and  high-fidelity decentralized deep learning framework.  To the best of our knowledge, \textbf{D$^{2}$-MHE} is  the first work to protect the privacy and accelerate the performance of decentralized learning systems using cryptographic primitives. Experimental results show that \textbf{D$^{2}$-MHE} can provide the optimal accuracy-performance trade-off compared to other state-of-the-art works.  In the future,  we will focus on improving the computation overhead of \textbf{D$^{2}$-MHE}, since this is the main bottleneck of the current homomorphic encryption applied to real-world applications.


\ifCLASSOPTIONcaptionsoff
\newpage \fi

\bibliographystyle{IEEEtran}
\bibliography{PPDR}
\begin{IEEEbiography}[{\includegraphics[width=1in,height=1.25in,clip,keepaspectratio]{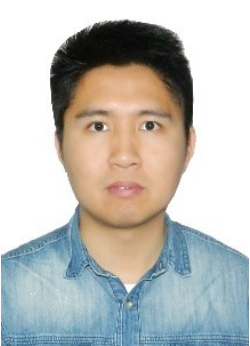}}]{Guowen Xu }
is currently a Research Fellow with Nanyang Technological University, Singapore. He received the Ph.D. degree at 2020 from University of Electronic Science and Technology of China. His research interests include applied cryptography and  privacy-preserving  issues in Deep Learning.
\end{IEEEbiography}
\vspace{-10 mm}
\begin{IEEEbiography}[{\includegraphics[width=1in,height=1.25in,clip,keepaspectratio]{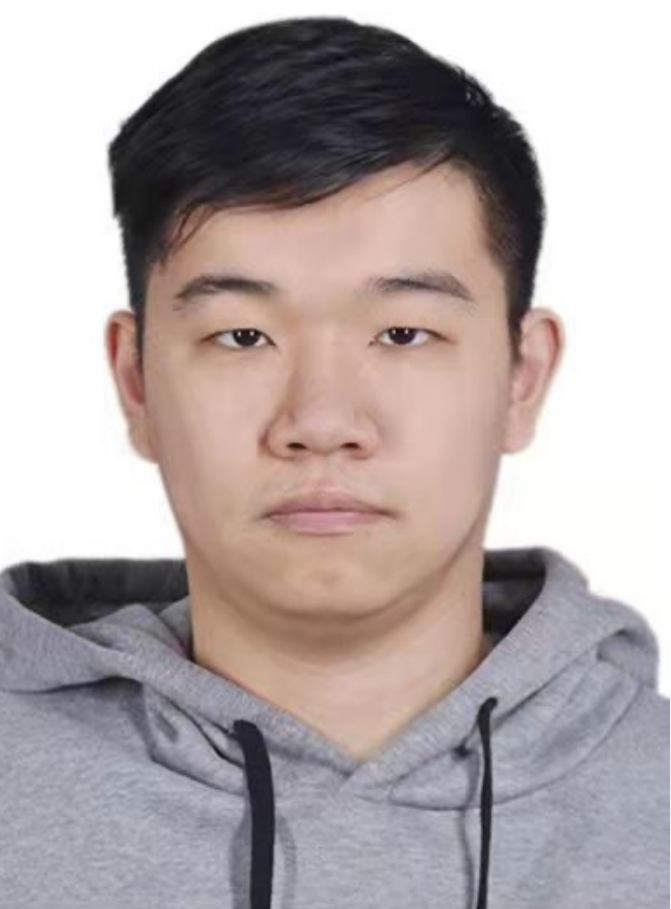}}]{Guanlin Li }
is currently a  Ph.D. student in the School of Computer Science and Engineering, Nanyang Technological University. He received the bachelor's degree in information security from the Mathematics School of Shandong University, Shandong, China in 2018. His research interests include deep learning, computer vision, adversarial example and neural network security.
\end{IEEEbiography}

\begin{IEEEbiography}[{\includegraphics[width=1in,height=1.25in]{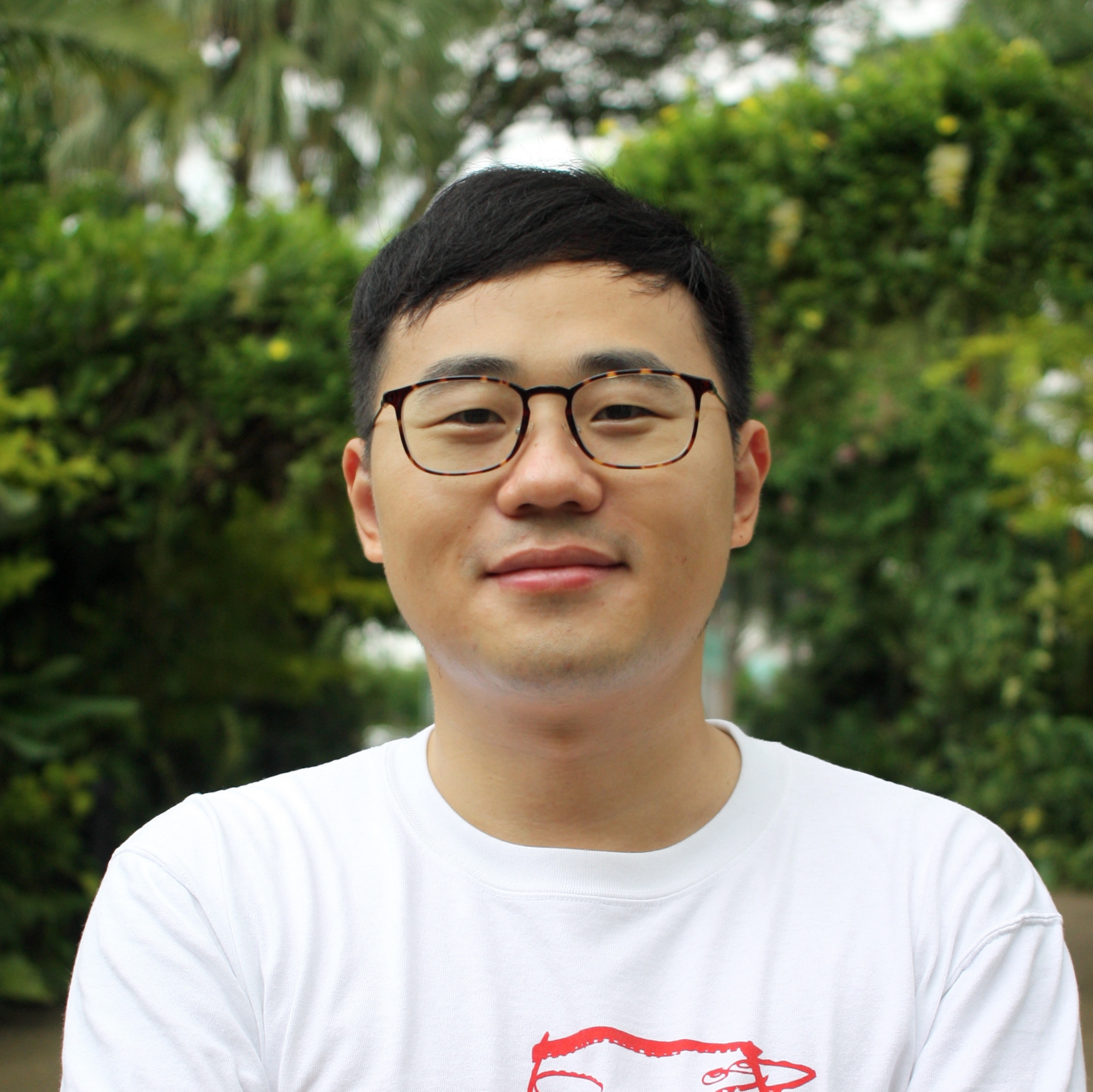}}] {Shangwei Guo} is an associate professor in College of Computer Science, Chongqing University. He received the Ph.D. degree in computer science from Chongqing University, Chongqing, China at 2017. He worked as a postdoctoral research fellow at Hong Kong Baptist University and Nanyang Technological
University from 2018 to 2020. His research interests include secure deep learning, secure cloud/edge computing, and database security.
\end{IEEEbiography}
\vspace{-10 mm}

\begin{IEEEbiography}[{\includegraphics[width=1in,height=1.25in]{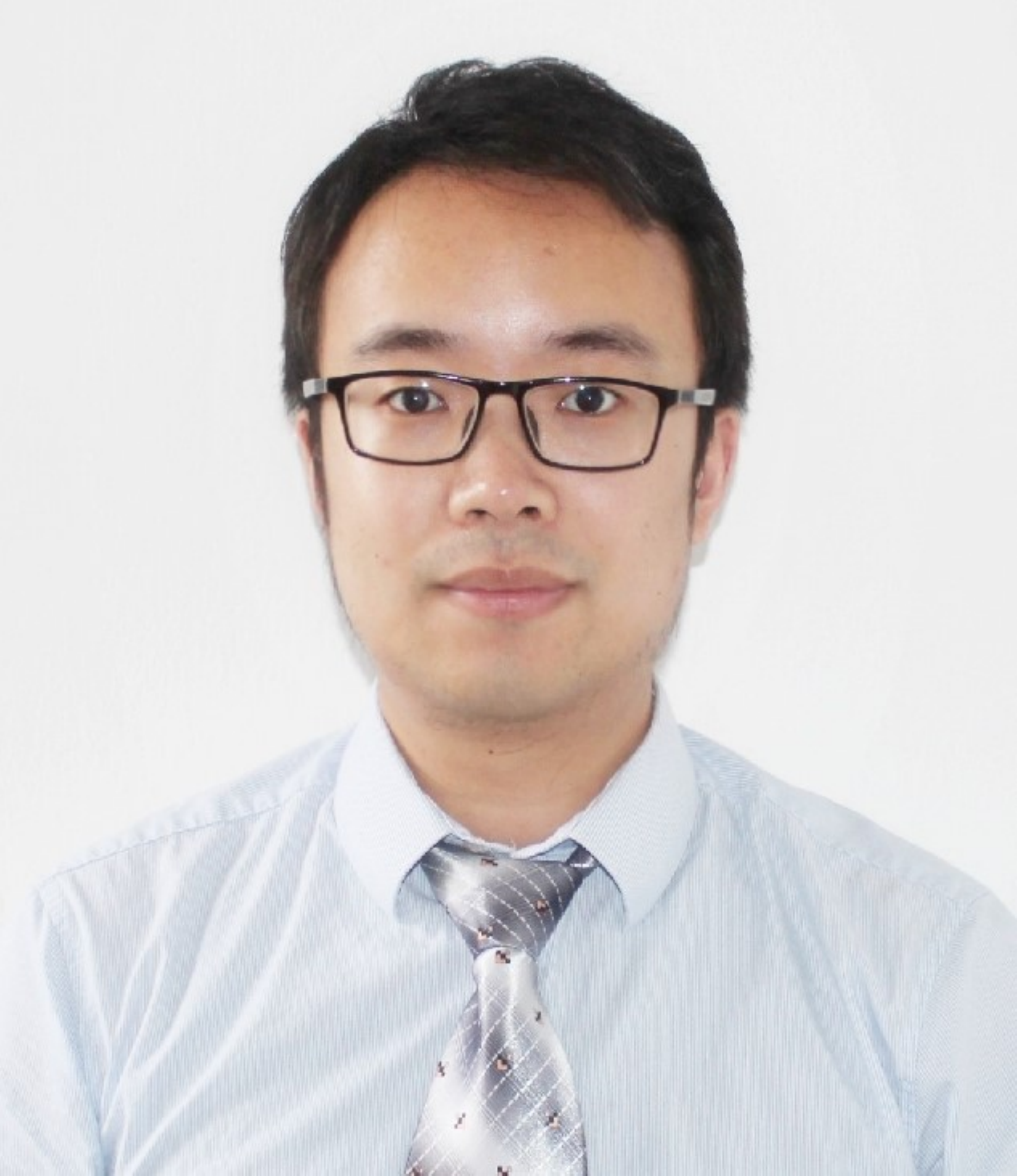}}] {Tianwei Zhang}is an assistant professor in School of Computer Science and Engineering, at Nanyang Technological University. His research focuses on computer system security. He is particularly interested in security threats and defenses in machine learning systems, autonomous systems, computer
architecture and distributed systems. He received his Bachelor's degree at Peking University in 2011, and the Ph.D degree in at Princeton University in 2017.
\end{IEEEbiography}
\vspace{-10 mm}
\begin{IEEEbiography}[{\includegraphics[width=1in,height=1.9in,clip,keepaspectratio]{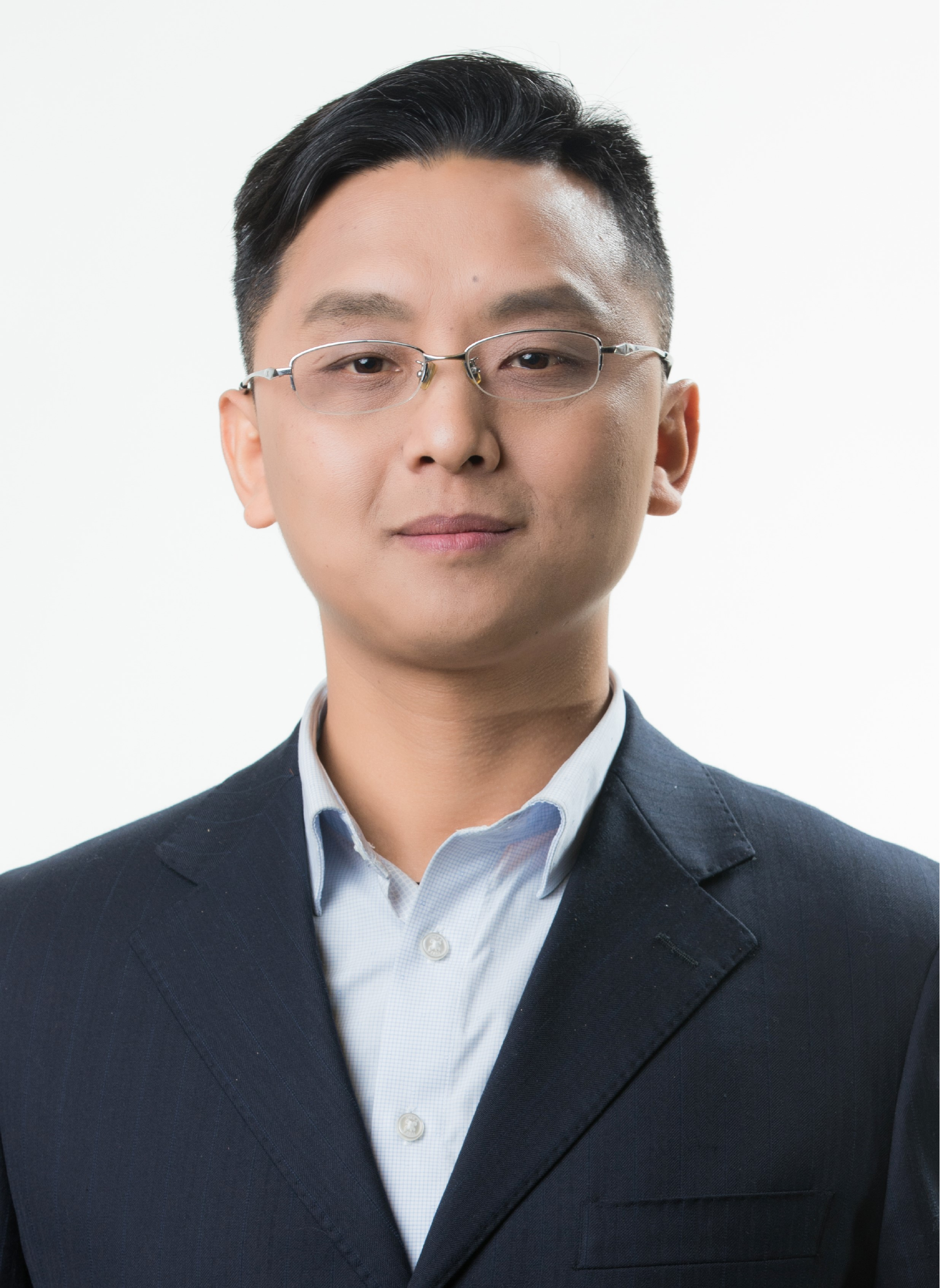}}]{Hongwei Li}
is currently the Head and a Professor at Department of Information Security, School of Computer Science and Engineering, University of Electronic Science and Technology of China.  His research interests include network security and applied cryptography. He is the Senior Member of IEEE, the Distinguished Lecturer of IEEE Vehicular Technology Society.
\end{IEEEbiography}
\vspace{-10 mm}

\end{document}